\documentclass[a4paper,11pt]{article}
\usepackage{amsmath, amssymb, amsthm}
\usepackage{bbm, natbib}
\usepackage{graphicx}
\usepackage{enumerate}
\usepackage[usenames]{color} 
\usepackage{exscale}
\usepackage[normalem]{ulem}

\definecolor{darkgreen}{rgb}{.2,.8,0}

\def\blue#1{\textcolor{blue}{#1}}
\def\red#1{\textcolor{red}{#1}}

\def\indf#1{\mathbbmss 1_{ \{ #1 \} }} 
\def\ind#1{\mathbbmss 1_{ #1  }} 
\def\TI{h}
\def\HX{\widehat X}
\def\sgn{\textrm{sgn}}
\def\cC{\mathcal{C}}
\def\M{\Xi}

\usepackage{geometry}
\newtheorem{theorem}{Theorem}[section]
\newtheorem{lemma}[theorem]{Lemma}
\newtheorem{proposition}[theorem]{Proposition}
\newtheorem{corollary}[theorem]{Corollary}

\theoremstyle{remark}
\newtheorem{remark}[theorem]{Remark}
\theoremstyle{definition}
\newtheorem{definition}[theorem]{Definition}
\newtheorem{assumption}[theorem]{Assumption}

\newcommand{\OldDelta}{\phi}

\renewcommand{\baselinestretch}{1}\normalsize

\newcommand{\commentOut}[1]{} 

\begin{document}

\title{Price manipulation in a market impact model with dark pool }

\author{Florian Kl\"{o}ck\thanks{University of Mannheim, Department of Mathematics, A5, 6, 68131 Mannheim, Germany, {\tt kloeck@uni-mannheim.de}}{\setcounter{footnote}{6}}\and Alexander Schied\thanks{University of Mannheim, Department of Mathematics, A5, 6, 68131 Mannheim, Germany, {\tt schied@uni-mannheim.de}}{\setcounter{footnote}{2}}\and Yuemeng Sun\thanks{Cornell University, ORIE, Rhodes Hall, Ithaca, NY 14850, U.S.A., {\tt  ys273@cornell.edu}\hfill\break
F.K. and A.S.  acknowledge support by Deutsche Forschungsgemeinschaft through Research Grant SCHI 500/3-1}}

\date{ $\,$}

\maketitle

\vspace{-0.5cm}

\begin{abstract}For a market impact model,  price manipulation and related notions play a role that is similar to the role of arbitrage in a derivatives pricing model.  Here, we give a systematic investigation into such regularity issues when orders can be executed both at a traditional exchange and in a dark pool.  To this end, we focus on a class of dark-pool models whose market impact at the exchange is described by an Almgren--Chriss model. Conditions for the absence of price manipulation for all Almgren--Chriss models include the absence of temporary cross-venue impact, the presence of full permanent  cross-venue impact, and  the additional penalization of    orders executed in the dark pool. When a particular  Almgren--Chriss model has been fixed, we show by  a number of examples that the regularity of the dark-pool model hinges in a subtle way  on the interplay of all  model parameters and on the liquidation time constraint. The paper can also be seen as a case study for the regularity of market impact models in general.
\end{abstract}


\section{Introduction}

Recent years have seen a mushrooming of alternative trading platforms called \emph{dark pools}.  Orders placed in a dark pool are not visible to other market participants (hence the name) and thus do not influence the publicly quoted price of the asset. Thus, when dark-pool orders are executed against a matching  order, no direct price impact is generated, although there may be certain indirect effects. 
Dark pools  therefore promise a reduction of  market impact and of the resulting liquidation costs. They are hence a popular platform for the execution of large trades. 

Dark pools differ from standard limit order books in that they do not have an intrinsic price finding mechanism. Instead, the price at which orders are executed is derived from the publicly quoted prices at an exchange. Thus, by manipulating the price at the exchange through placing suitable buy or sell orders, the value of the \lq\lq dark liquidity" in the dark pool can  be altered. For this reason, dark pools have drawn significant attention by regulators; see \citet{iosco11}. We refer to \citet{mittal2008} and  for a practical overview on dark pools and some related issues of market manipulation. More information on the practical background of our analysis is given in the excellent survey article \citet{Lehalle} and the recent book \citet{LehalleBook}.

Our first main goal in this paper is to conduct a systematic mathematical analysis of  market manipulation  by means of a dark pool. To this end, we study the mathematical concept of \emph{price manipulation} as introduced by \citet{hubermanstanzl} as well as some other, related regularity concepts. \citet{hubermanstanzl} introduced the notion of price manipulation in the framework of a generic market impact model. They argue that by  repeating suitably rescaled price manipulation strategies one can generate so-called quasi-arbitrage, which can be regarded as some kind of statistical arbitrage.  In this sense, the absence of price manipulation can be viewed as a specific no-arbitrage condition for market impact models. Its existence can indicate problems with the model at hand; for instance, the optimal trade execution problem may not be well-posed, because optimal strategies may not exist.  Our main results will show in particular that, for our dark-pool model,  price manipulation can only be excluded by rather unrealistic choices of  model parameters. This happens despite the fact that we choose a rather restrictive class of strategies, which excludes the most notorious predatory strategies such as \lq\lq fishing" \citep{mittal2008}. We believe that our findings can be interpreted so as to  provide some support for the general concerns expressed over market manipulation possibilities with dark pools \citep{iosco11, mittal2008}.

Our second main goal is to provide a case study for the regularity analysis of market impact models in general. As explained above, the notion of price manipulation for a market impact model has been proposed as some analogue of the notion of arbitrage for an asset price model. But while  arbitrage and many related concepts  have been studied throughout decades in the context of asset pricing models,  the study of the regularity or irregularity of market impact models is just at its beginning; see \citet{GatheralSchiedSurvey} for a recent survey.  We therefore investigate  here several additional regularity conditions besides the absence of price manipulation. These notions include transaction-triggered price manipulation as introduced in \citet{alfonsischiedslynko2010}, and the new concept of nonnegative expected liquidation costs, which was introduced independently by \citet{rochsoner}. In particular, in Section \ref{SpecResSec} we will see that for certain choices of model parameters some of  these notions may be violated 
while others are not, and the situation can actually become relatively complex. We hope that our analysis  can thus provide some additional insight into the various existing notions of regularity and their relations. Further studies of the regularity of market impact models are given in \citet{Gatheral}, \citet{AS}, \citet{alfonsischiedslynko2010}, \citet{AlfonsiAcevedo}, and \citet{Kloeck}, to mention only a few.

We use a stochastic model for trade execution at two possible venues: a dark pool and an exchange.   It is a natural model, because it  extends the  standard  Almgren--Chriss market impact model for  exchange  prices \citep{BertsimasLo,almgrenchriss2001,Almgren2003} by a generic 
dark pool with a fairly general flow of incoming orders.  It is possible to understand this  dark pool  as the aggregation of all available dark pools, when there are several such venues in the market.  
The problem of optimally routing dark-pool orders within a collection of dark pools has been analyzed in  \citet{laruelle}. 

The optimal trade execution problem for orders placed in a dark pool and at an exchange has been considered before by \citet{kratzschoeneborn}. Their model, however, only allows for quadratic transaction costs for orders placed at the exchange. In particular   \citet{kratzschoeneborn} exclude any temporary or transient price impact components. 
This exclusion of genuine price impact components is not just made out of mathematical convenience: Our results show that any inclusion of price impact will lead to the existence of profitable price manipulation strategies unless the remaining model parameters are chosen in an extreme manner, and these price manipulation strategies will have a significant impact on the possible existence and structure of optimal execution strategies. In particular the simple recipes found by \citet{kratzschoeneborn} will no longer work in a genuine market impact model; see also Propositions 7.1 and 7.2 in  \citet{kratzschoeneborn}.\citet{KratzSchoeneborn12} have recently extended their previous results to   continuous time, but still exclude genuine price impact components from their model.

 In Section \ref{GenResSection}, our first main result characterizes completely those models from our class that are sufficiently regular for all underlying Almgren--Chriss models, either  in the sense of the absence of price manipulation or in terms of the new condition of   nonnegative expected liquidation costs.  The critical quantities will be the size of \lq\lq slippage" and the degrees of permanent and temporary cross-venue impact. In Section \ref{SpecResSec}, we then investigate  the existence of model irregularities  for special model characteristics. It will turn out that the generation of such irregularities hinges in a subtle way  on the interplay of all  model parameters and on the liquidation time constraint.  In Section \ref{section_optliq} we illustrate in a simplified setting that our regularity condition guarantees the existence of optimal trade execution strategies, and we show how such strategies can be computed.

The paper is organized as follows. In the subsequent Section \ref{ModelSection} we introduce the model and formulate our standing assumptions. In Section \ref{PriceManipulationSection} we review and discuss several notions for the regularity of a market impact model, namely the absence of standard and transaction-triggered price manipulation and the new condition of nonnegative expected liquidation costs. Our main results are stated in Section \ref{MainResultsSec} and proved in Section \ref{ProofsSection}. We conclude in Section \ref{ConclusionSection}.

\section{Model setup}\label{ModelSection}

In Section \ref{AlmgrenChrissSection} we recall the basic notions of the well-known Almgren--Chriss model \citep{BertsimasLo,almgrenchriss2001} in the form of \citet{Almgren2003} and introduce (fairly general) conditions we will impose on its parameters. In Section \ref{AC-darkpool Section} we will explain how a dark pool can be added to this model and how admissible trading strategies can be defined. In Section \ref{Costs Section} we will define the liquidation costs associated   with an admissible trading strategy. The total costs will comprise some obvious expense terms and certain additional, indirect costs.  As explained in the Introduction, 
 \citet{kratzschoeneborn,KratzSchoeneborn12} have proposed a model for trading in a dark pool and at an exchange but they force their model to be free of price manipulation strategies by excluding  all permanent or transient price impact components and only consider quadratic transaction costs. One of our main goals here is to understand the effects that arise when a genuine market impact model is used to describe the price impact of orders executed at the exchange.

\subsection{The Almgren--Chriss model}\label{AlmgrenChrissSection}

In the Almgren--Chriss market impact model, it is assumed that the number of shares in the trader's portfolio  is described by an  absolutely continuous trajectory $t\mapsto X_t$, the trading strategy. 
Given this trading trajectory, the price at which transactions occur is
\begin{equation} \label{price_exg}
	 P_t = P^0_t + \gamma (X_t - X_0) + \TI (\dot X_t).
\end{equation}
Here, $P^0_t$ is the \emph{unaffected stock price process}.  The term $\gamma (X_t-X_0)$ corresponds to the \emph{permanent price impact} that has been accumulated by all transactions until time $t$. It is usually assumed to be linear in $X_t-X_0$ with $\gamma$ denoting a positive constant, because  linearity  guarantees the absence of price manipulation; see \citet{hubermanstanzl} or \citet{Gatheral}.  See also \citet{AlmgrenHauptmanLi} for an empirical analysis and justification of this assumption. 
The term $\TI (\dot X_t)$ describes  the \emph{temporary} or \emph{instantaneous impact} of trading $\dot X_t\,dt$ shares at time $t$ and only affects this current order. It can also comprise other direct or indirect costs, which are often called \lq\lq slippage"; see Section 2.2 in \citet{Gatheral}. The most common choices in the literature  are 
 linear impact, $h(x)=\eta x$, or more generally  power-law impact of the form $
h(x)=\eta \,\sgn(x)|x|^\nu$ for two constants $\eta\ge0$ and $\nu>0$.  Here we will consider  general functions $h$ that are subject to the following definition. 

\begin{definition}\label{temporary impact def}A function $\TI: \mathbb R \rightarrow \mathbb R$ will be called an \emph{impact function} if it is continuous, strictly increasing,  satisfies $h(0)=0$, and is such that $f(x)=xh(x)$ is convex. 
\end{definition}

\begin{assumption}\label{AlmgrenChrissAssumption}
We assume that the unaffected stock price process $(P^0_t)_{t\geq 0}$  is a c\`adl\`ag martingale   on a filtered probability space $(  \Omega, {\mathcal F}, ({\mathcal F_t}), {\mathbb P})$ for which $\mathcal F_0$ is $\mathbb{P}$-trivial. Moreover, the permanent-impact parameter $\gamma$ is  strictly positive, and $h$ is an impact function.  An Almgren--Chriss model is thus defined in terms of the parameters
$(\gamma, h, P^0)$.
\end{assumption}

The condition that $P^0$ is a martingale is a standard assumption in the market impact literature. One  reason is that drift effects can be ignored due to the usually short trading horizons. In addition, we are interested here in the qualitative effects of  price impact on the stability of the model. A nonzero drift  would lead to the existence of  profitable \lq\lq round trips" that would have to be distinguished from price manipulation strategies in the sense of \citet{hubermanstanzl} (see  Definition \ref{PMDef} below).

\subsection{An Almgren--Chriss model with dark pool}\label{AC-darkpool Section}

The Almgren--Chriss model is a market impact model for exchange-traded orders. We will now extend this model by allowing the additional execution of orders in a \emph{dark pool}. A dark pool is an alternative trading venue in which unexecuted orders are invisible to all other market participants. In this dark pool, buy and sell orders are matched and executed at the current price at which the asset is traded at the exchange.

 In addition to a trading strategy   executed at the exchange, investors can place  an order of $\HX$ shares into the dark pool at time $t=0$. This order will be matched with incoming orders of the opposite side. These orders arrive at random times $0< \tau_1 < \tau_2 < \ldots$ and we denote the size of incoming matching orders by $\widetilde Y_1, \widetilde Y_2, \ldots>0$. We  consider only those orders that are a possible match.  That is, the $\widetilde Y_i$ will describe sell orders when $\widehat X>0$ is a buy order and buy orders when $\widehat X<0$ is a sell order. In addition, the investor can set a \emph{minimum quantity level}\footnote{Restricting the size of a matching order by setting a minimum quantity level $M$ is a common feature in many real-world dark pools. For instance, it helps to protect against \lq\lq fishing"\ by predatory traders; see \citet{mittal2008}.} $\M\in[0,|\widehat X|]$.  Then all incoming orders $\widetilde Y_i$ of size $\widetilde Y_i\ge \M$ will  be matched piece by piece with the order $\HX$ until it is cancelled or completely filled.  That is,
\[ Y_i := \begin{cases}
\sgn(\widehat X)\widetilde Y_i, & \textrm{ if } \widetilde Y_i\ge \M\text{ and }\sum_{j=1}^{i-1} Y_j + \widetilde Y_i\le |\HX|, \\
\HX -\sgn(\widehat X) \sum_{j=1}^{i-1}  Y_j, & \textrm{ if } \widetilde Y_i\ge \M\text{ and } \sum_{j=1}^{i-1}  Y_j \le |\HX| \textrm{ and } \sum_{j=1}^{i-1} Y_j +\widetilde Y_i> |\HX|,\\
0, & \textrm{ otherwise, } 
\end{cases}
\]
is  the part of the incoming  order that is actually executed against the remainder of $\widehat X$.
By defining the counting process associated with the arrival times $(\tau_k)$,
\begin{equation}\label{NtDefEq}
N_t:=\max\{k\in\mathbb N\,|\,\tau_k\le t\},
\end{equation}
 the amount of shares that have been executed in the dark pool until time $t$ can be conveniently denoted by
\begin{equation}\label{Zt def eq}
Z_t: = \sum_{i=1}^{N_t} Y_i.
\end{equation}
 By $(\mathcal G_t)$ we denote the right-continuous filtration generated by $(\mathcal F_t)$, $Z$, and $(\tau_i\wedge t)_{t\ge0}$ for $i=1,2,\dots$.

In the first part of the paper, we make some very mild assumptions on the laws and interdependence of the random variables $(\tau_i)$, $(\widetilde Y_i)$, and $P^0$:

\begin{assumption}\label{DarkPoolAssumption1}We assume the following  conditions:
\begin{enumerate}
\item $0<\tau_1<\infty$ $\mathbb{P}$-a.s., and $\displaystyle\lambda_0:=\inf_{0<\delta\le1}\frac{1}{\delta}\mathbb{P}[\,\tau_1\le\delta\,]>0$.
\item For all $x>0$ we have $\displaystyle \lambda_1(x):=\inf_{\delta>0}\mathbb{P}[\,\widetilde Y_1\ge x\,|\,\tau_1\le\delta\,]>0$.
\item $P^0$ is a martingale also under the filtration $(\mathcal{G}_t)$ generated by $(\mathcal F_t)$, $Z$, and $(\tau_i\wedge t)_{t\ge0}$ for $i=1,2,\dots$. 
\end{enumerate}
\end{assumption}

Condition (a) means that the intensity for the arrival of the first matching order is bounded away from zero. Condition (b) states that for all $x>0$ there is a positive probability that the first incoming matching order has at least size $x$, conditional on the event that $\{\tau_1\le\delta\}$. This latter condition can be relaxed, but we use it in this form to keep our assumptions simple.
The requirement that $P^0$ is  a $(\mathcal{G}_t)$-martingale allows $(\tau_i)$ and $(\widetilde Y_i)$ to depend on $P^0$ in an arbitrary manner but, conversely,  limits the dependence of $P^0$ on these random variables. This limitation is entirely natural, since below we will explicitly model the indirect cost impact of fills in the dark pool on order executions at the exchange  via \eqref{costs3}.  Note that all conditions in Assumption \ref{DarkPoolAssumption1} are satisfied in particular when $\tau_1$ has an exponential distribution, $(\tau_i)$, $(\widetilde Y_i)$, and $P^0$ are independent random variables, and $\widetilde Y_1$ is unbounded from above, as we will assume in the second part of the paper.

\bigskip

Now we consider an investor who must liquidate an initial asset position of $X_0 \in \mathbb R$ shares during the time interval $[0,T]$.  The problem of how to do this in an optimal fashion is known as the optimal trade execution problem;  see, e.g., \citet{GokayRochSoner}, \citet{Lehalle}, \citet{GatheralSchiedSurvey}, and the references therein.

In the extended dark pool model, the investor will first place an order of  $\HX \in \mathbb R$ shares  in the dark pool\footnote{If at time $t=0$ the dark pool contains an order $\widetilde Y_0$ of the opposite side, then the investor could fill this order immediately and then start liquidating the remaining asset position $X_0-\widetilde Y_0$, maybe by resizing the dark-pool order. Therefore we can assume  that  the dark pool does not contain a matching order at $t=0$. Moreover, restricting the placement of dark-pool orders to $t=0$ lets us exclude so-called \lq fishing\rq\ strategies; see Remark \ref{FishingRemark} below.}    and then choose a liquidation strategy of Almgren--Chriss-type for the execution of the remaining assets at the exchange. This latter strategy must be absolutely continuous in time. It will thus be described by a  process $(\xi_t)$ that parameterizes the speed by which shares are sold at the exchange. Moreover, until fully executed,  the remaining part of the order $\widehat X$ can be cancelled at a (possibly random)  time $\rho < T$. Hence, the number of shares held by the investor at time $t$ is
\begin{equation}\label{StrategyEq}
 X_t := X_0 + \int_0^t \xi_s \, ds + Z^\rho_{t-}, \end{equation} 
where $Z^\rho_{t-}$ denotes the left-hand limit of $Z^\rho_t=Z_{\rho\wedge t}$.

\begin{definition}\label{strategy chi def} Let an initial position $X_0\in\mathbb R$ and a liquidation horizon $T>0$ be given.
An \emph{admissible trading strategy} is a  quadruple $\chi:=(\HX, \M, \xi, \rho)$ where $\HX \in \mathbb R$  is the size of the dark-pool order, $\M\in[0,|\widehat X|]$ is the minimum quantity level,  the cancellation time $\rho$ is a $(\mathcal G_t)$-stopping time  such that $\rho < T$ $\mathbb P$-a.s., and  the trading strategy $\xi$ is a $(\mathcal G_t)$-predictable process whose integral, $\int_0^t\xi_s\,ds$, is $\mathbb P$-a.s. bounded uniformly in $t$ and $\omega$. In addition, the liquidation constraint
\begin{equation}\label{liquidationconstraint}
X_0 + \int_0^T \xi_t \, dt + Z_\rho =0
\end{equation} 
must be $\mathbb P$-a.s. satisfied.
The set of all admissible strategies for  given $X_0$ and $T$ is denoted by $\mathcal X(X_0,T)$. 
\end{definition}

Due to \eqref{StrategyEq} and \eqref{liquidationconstraint},  the terminal asset position of any admissible strategy is $X_T=0$, since our requirement $\rho<T$ implies that $Z^\rho_{T-}=Z_\rho$.

Suppose that the admissible trading strategy $\chi=(\HX, \M, \xi, \rho)$ is used. As in \eqref{price_exg}, the price at which assets can be traded at the exchange is defined as
\begin{equation}\label{AffectedPriceProcessExchange}
P^\chi_t = P^0_t + \gamma \int_0^t \xi_s \, ds + \TI(\xi_t).
\end{equation}
The price at which the $i^{\text{th}}$ incoming order  is executed in the dark pool will be 
\begin{equation}\label{hat P tau}
\widehat P^\chi_{\tau_i}  =  P^0_{\tau_i} + \gamma \int_0^{\tau_i} \xi_s \, ds   + \kappa h(\xi_{\tau_i}),
\end{equation}
where
$$\kappa\ge0.
$$
That is, we take as price  $P^\chi_{\tau_i}$ from \eqref{AffectedPriceProcessExchange} but allow the temporary price impact, $\kappa h(\xi_{\tau_i})$, to be different from $h(\xi_{\tau_i})$. This assumption is natural for at least two reasons. First, as  explained above and in Section 2.2 of \citet{Gatheral}, the temporary price impact is often used to describe also direct or indirect transaction costs (\lq\lq slippage") that do not arise from a change in the exchange-quoted price. When this is the case, \eqref{AffectedPriceProcessExchange} will be a virtually adjusted price that does not fully coincide with the actual exchange-quoted price, and so the price $\widehat P^\chi_{\tau_i}$ for dark-pool executions needs to be adjusted correspondingly. Second, and more importantly, the continuous-time Almgren--Chriss model is an approximation to a discrete-time reality in which the temporary price impact describes that part of the price impact of a child order which has already decayed by the time the next child order is placed \cite[Section 1.3]{almgrenchriss2001}. So, unless an execution in the dark pool happens  instantaneously after a child order is placed at the exchange, the price of the dark-pool execution will not receive the full temporary price impact generated by the preceding child order. Thus, for both reasons explained above, it is natural to allow that 
only a fraction of the temporary price impact component $h(\xi_t)$ is affecting the price at which orders in the dark pool are executed.  This argument suggests to take $\kappa\in[0,1]$. The mathematical arguments, however, work for general $\kappa\ge0$, and so we will not require $\kappa\le1$ in the sequel.

\subsection{The  costs associated with an admissible strategy}\label{Costs Section}

Let us now introduce the  \emph{costs} associated with an admissible trading strategy $\chi=(\HX,\Xi, \xi, \rho) \in\mathcal{X}(X_0,T)$. The total costs will consist of several components.  The first component is given by 
\begin{equation}\label{costs1}
\int_0^T\xi_tP^\chi_t\,dt.
\end{equation}
This term is standard within the Almgren--Chriss setting and simply describes the accumulated expenses from  buying $\xi_t\,dt$ shares at  price $P^\chi_t$ at each time $t$. Similarly, the next term describes the accumulated expenses from buying $Y_i$ shares at price $\widehat P^\chi_{\tau_i}$ whenever there is a fill in the dark pool:
\begin{equation}\label{costs2}
\sum_{i=1}^{N_{\rho}}Y_i\widehat P^\chi_{\tau_i}.
\end{equation}
Besides these expenses, there can also be certain indirect costs associated with the usage of a dark pool.  First, when our order  $\widehat X$ gets a fill in the dark pool, the counterparty's matching order   takes our liquidity from the dark pool. Without the presence of our liquidity in the dark pool, the counterparty may have executed at least part of their order at the exchange, where it would have generated---possibly favorable---permanent price impact. We therefore assume that each $Y_i$ would have generated the permanent price impact $-\alpha \gamma Y_i$ at the exchange, where $\alpha\in[0,1]$ describes the fraction of $Y_i$ that would have been executed at the exchange had our order $\widehat X$ not been placed.  This impact would have influenced all future prices at the exchange and, in turn, also at the dark pool.  This reasoning leads to the definition of the following indirect cost component:
\begin{equation}\label{costs3}
\alpha\gamma\Big(\int_0^T Z^\rho_{t}\xi_t\,dt+\sum_{i=1}^{N_{\rho\wedge T}}Z_{\tau_i-}Y_i\Big).
\end{equation}
Finally, we propose to take into account additional \lq\lq slippage"\ generated by an execution in the dark pool. Besides actual transaction costs or taxes, this slippage may comprise hidden costs that relate to  dark pool executions and that are extremely difficult to model explicitly. For instance, one can think of costs arising from the phenomena of adverse selection or \lq\lq fishing"; see \citet{mittal2008} and Remark \ref{FishingRemark}. Moreover, due to the covert nature of dark pools, data may be sparse so that there will be a high degree of model uncertainty. In applied financial engineering, a slippage cost term can thus also serve as a penalization of dark-pool orders in view of possible  model misspecification. We assume that for each execution $Y_i$ the resulting cost component is given by $\beta(Y_i)Y_i$, where $\beta$ is an impact function in the sense of Definition \ref{temporary impact def}. This leads to a cost term  the form
\begin{equation}\label{costs4}
\sum_{i=1}^{N_{\rho}}\beta(Y_i)Y_i.
\end{equation}

So to speak, the factor  $\alpha$ models the \lq\lq virtual permanent price impact" generated by a dark pool execution, whereas the impact function $\beta$ describes the corresponding \lq\lq temporary price impact" effect.
Let us summarize the assumptions and definitions we made so far:

\begin{definition}
 The \emph{dark-pool extension} of a given Almgren--Chriss model $(\gamma, h,P^0)$ is defined in terms of the new parameters
$$
(\alpha,\beta,\kappa,(\tau_i),(\widetilde Y_i)),
$$
where $\alpha\in[0,1]$, $\beta$ is an impact function, $\kappa\ge0$, and the sequences $(\tau_i)$ and $(\widetilde Y_i)$ satisfy Assumption \ref{DarkPoolAssumption1}. In the extended model, the \emph{costs} associated with an admissible strategy $\chi=(\HX, \M, \xi, \rho)$ are given by
\begin{equation}\label{cost def}
\cC_T^\chi=X_0P_0^0+\int_0^T\xi_tP^\chi_t\,dt+\alpha\gamma\int_0^T Z^\rho_{t}\xi_t\,dt+\sum_{i=1}^{N_{\rho}}Y_i\Big(\widehat P^\chi_{\tau_i}+\alpha\gamma Z_{\tau_i-}+\beta(Y_i)\Big).
\end{equation}
\end{definition}

\bigskip

In \eqref{cost def}, the costs  $\cC_T^\chi$ are defined as the sum of the initial face value, $X_0P_0^0$, and the expense  and cost terms \eqref{costs1}---\eqref{costs4}. So, for a short position $X_0<0$,  $\cC_T^\chi$ is the amount by which the actual expenses exceed the amount $-X_0P_0^0$. Similarly, for a long position $X_0>0$,  $\cC_T^\chi$ is the amount by which the
revenues of the strategy $\chi$ fall short of the face value  $X_0P_0^0$. When it is clear from the context which strategy $\chi=(\HX, \Xi,\xi, \rho)$ is used, or when considering generic strategies, we often simply write $\cC_T$ instead of $\cC_T^\chi$.

\section{Price manipulation}\label{PriceManipulationSection}

Our first main concern in this paper is to investigate the stability and regularity of the dark-pool extension in dependence on the parameters $(\gamma, h, P^0)$ and $(\alpha,\beta,g,(\tau_i),(\widetilde Y_i))$.  This question is analogous to establishing the absence of arbitrage in a derivatives pricing model, where absence of arbitrage is a necessary condition for the pricing  of  contingent claims by replication.

But there must also be a difference in the notions of  regularity of a derivatives pricing model and of a market impact model. In a derivatives pricing model, one is interested in constructing strategies that almost surely replicate a given contingent claim, and this is the reason why   one must exclude the existence of arbitrage opportunities defined in the usual almost-sure sense. In a market impact model, one is interested in constructing optimal trade execution strategies. These strategies are not defined in terms of an almost-sure criterion but as minimizers of a cost functional of a risk averse investor.  Commonly used cost functionals involve expected value as in \citet{BertsimasLo} and \citet{Gatheral}, mean-variance criteria as in \citet{almgrenchriss2001}, expected utility as in \citet{schiedschoeneborn2009}, or alternative risk criteria as in  \citet{Forsythetal} and \citet{GatheralSchied}. Therefore, also the regularity conditions to be imposed on a market impact model need to be formulated in a similar manner. To make such regularity conditions independent of particular investors preferences, it is  reasonable to formulate them in a risk-neutral manner. The first condition that was proposed in this context is the following.

\begin{definition}[\citet{hubermanstanzl}]\label{PMDef} A \emph{round trip} is an admissible trading strategy with $X_0 = 0$.
A \emph{price manipulation strategy} is a round trip that has strictly negative expected costs, $\mathbb E [\, \mathcal C_T \,]<0$.
\end{definition}

Of course, costs can be regarded as negative revenues and vice versa. Thus, the existence of price manipulation strategies allows to make a positive expected profit by exploiting one's own price impact. \citet{hubermanstanzl} argue that by rescaling and repeating  price manipulation strategies in their model, one can generate so-called quasi-arbitrage, which can be viewed as some kind of statistical arbitrage. Conversely, any quasi-arbitrage strategy yields a price manipulation strategy. In this sense, the absence of price manipulation can be regarded as a specific no-arbitrage condition for market impact models. Moreover, when the costs are a convex functional of an trade execution strategy, as it is often the case,  the existence of price manipulation precludes the existence of optimal trade execution strategies for risk-neutral investors, because one can generate arbitrarily large expected revenues by adding a multiple of a price manipulation strategy.  One can see easily that  in most cases the same argument also works for risk-averse investors provided that   risk aversion is small enough. \citet{kratzschoeneborn} mention the problem of price manipulation in a dark-pool model and give some preliminary results but no systematic analysis. 
Analyses of the absence of price manipulation in various other market impact models were given by \citet{hubermanstanzl}, \citet{Gatheral}, \citet{AS}, \citet{alfonsischiedslynko2010},  \citet{AlfonsiAcevedo}, and \citet{Kloeck}, to mention only a few; see  \citet{GatheralSchiedSurvey} for a recent overview. 

It was observed in \citet{alfonsischiedslynko2010} that the absence of price manipulation may not be sufficient to guarantee the stability of the model, because optimal trade execution strategies can still oscillate strongly between alternating buy and sell trades, a property one should exclude for various reasons. For instance, such trading strategies may be considered illegal when performed by a broker who is executing  a trade on behalf of a client. See also the discussions in \citet{alfonsischiedslynko2010} and \citet{GatheralSchiedSurvey}.  This was the reason for introducing in \citet{alfonsischiedslynko2010} the  notion of transaction-triggered price manipulation, which informally means that the expected execution costs of a sell (buy) program can be decreased by intermediate buy (sell) trades. In our specific situation, this idea can be made precise as follows. 

\begin{definition}There is \emph{transaction-triggered price manipulation} if there exists $X_0\in\mathbb R$, $T>0$, and a strategy $(\HX, \M, \xi, \rho) \in\mathcal{X}(X_0,T)$ for which either $\HX$ or some $\xi_t$ have the same sign as $X_0$ and that has strictly higher expected revenues than all strategies $(\HX', \M', \xi', \rho') \in\mathcal{X}(X_0,T)$ for which both  $\HX'$ and $\xi'_t$ have always the opposite sign of $X_0$. \end{definition}

  We will also consider the following notion of regularity, which was introduced independently in \citet{rochsoner}:

\begin{definition}The model has \emph{nonnegative expected liquidation costs} if for all $X_0\in\mathbb R$, $T>0$, and every  admissible   strategy $\chi \in\mathcal{X}(X_0,T)$,
\begin{equation} \label{illiquid-cond}
 \mathbb E[\mathcal C^\chi_T] \geq 0 .
\end{equation}
\end{definition}

Condition \eqref{illiquid-cond} states that on average it is not possible to make a profit beyond the face value of a position out of the market impact generated by one's own trades. We have the following hierarchy of regularity conditions in our model.

\begin{proposition} \label{hierarchy-lemma}
\begin{enumerate}
\item If there is no transaction-triggered price manipulation then the model has nonnegative expected liquidation costs.
\item If the model has nonnegative expected liquidation costs, then there is no price manipulation.
\end{enumerate}
\end{proposition}

Implication~(a) holds for every market impact model in which buy orders increase the price and sell orders decrease the price.  Implication~(b) clearly holds for every market impact model.

\begin{remark}\label{FishingRemark} A common price manipulation strategy is the so-called \lq\lq fishing" strategy in dark pools; see \citet{mittal2008}. In a fishing strategy,  agents first send  small orders  to   dark pools so as to detect dark liquidity. Once a dark-pool order is detected, the visible price at the exchange is manipulated for a short period in a direction that is unfavorable for that order. Finally, an order is sent to the dark pool so as to be executed against the dark liquidity at the manipulated price.

 Here, we are not interested in the profitability of such predatory fishing strategies but primarily in the stability and regularity of optimal trade execution algorithms in dark pool and exchange. We therefore exclude fishing strategies by allowing the placement of orders in the dark pool only at time $t=0$.
Allowing for the placement of dark-pool orders at times $t>0$ will lead to a larger class $\widetilde {\mathcal X}(X_0,T)\supset\mathcal X(X_0,T)$  of admissible strategies. Since clearly
$$\inf_{\chi \in\widetilde{\mathcal X}(X_0,T)}\mathbb E[\,\mathcal C^\chi_T\,]\le\inf_{\chi\in\mathcal X(X_0,T)}\mathbb E[\,\mathcal C^\chi_T\,],$$
 the conditions of no-price manipulation or of nonnegative expected liquidation costs will be violated as soon as they are 
violated in our present setting.  But one of the main economic consequences of our main results, as stated in the subsequent section, is that these conditions are typically  violated  for the smaller class $\mathcal X(X_0,T)$ when model parameters are realistic.  Therefore, they are also typically violated for any larger  class of strategies.
\end{remark}

\section{Results}\label{MainResultsSec}

An Almgren--Chriss model is specified by the parameters $(\gamma, h, P^0)$ satisfying Assumption~\ref{AlmgrenChrissAssumption}. Its extension  incorporating a dark pool is based on the additional parameter set $(\alpha,\beta,\kappa,(\tau_i),(\widetilde Y_i))$, which will always be assumed to satisfy Assumption \ref{DarkPoolAssumption1}.  We are interested in the conditions we need to impose on 
these parameters such that the extended market model is regular. Here, regularity refers to the absence of price manipulation and related notions as explained in the preceding section. 

\subsection{General regularity results}\label{GenResSection}

Our first result characterizes completely those parameters $ (\alpha,\beta,g,(\tau_i),(\widetilde Y_i))$ for which the dark-pool extension of \emph{every} Almgren--Chriss model is sufficiently regular for all time horizons.

\begin{theorem} \label{model-parameters}For given $(\alpha,\beta,\kappa,(\tau_i),(\widetilde Y_i))$, the following conditions are equivalent.
\begin{enumerate}
\item For any Almgren--Chriss model and all $T>0$, the dark-pool extension has nonnegative expected liquidation costs.
\item For any Almgren--Chriss model and every time horizon $T>0$,  the dark-pool extension does not admit price manipulation.

\item We have $\alpha = 1$, $\kappa =0$, and \begin{equation}\label{Delta>0 for alpha=1 Condition}
|\beta(y)|\ge\frac\gamma2|y|\qquad \text{for all $y\in\mathbb R$.}
\end{equation}
\label{thm41point-c}
\end{enumerate}
\end{theorem}
\goodbreak

\begin{remark}\label{ConditionsRemark}
Let us comment on the three conditions in part \commentOut{Report 2, minor comment 5}\ref{thm41point-c} of the preceding theorem.
\begin{enumerate}
\item[(i)]
The requirement $\alpha=1$ means that an execution of a dark-pool order must generate the same  virtual permanent impact on the exchange-quoted price as a similar order that is executed at the exchange. 
In view of the discussion preceding \eqref{costs4},  virtual price impact generated by the execution of a dark-pool order can be understood in terms of  a deficiency in opposite-price impact.

\item[(ii)] The requirement $\kappa=0$ means that temporary impact from trades executed at the exchange must not affect the price at which dark-pool orders are executed. Note that our notion of  admissibility of strategies excludes short-term manipulation of the exchange-quoted price in  simultaneous response of the arrival of a matching order in the dark pool, because  $(\mathcal G_t)$-predictability basically allows $\xi_t$ to depend on $Y_{\tau_i}$ only for $t>\tau_i$. Therefore the requirement $\kappa=0$ is quite surprising.  Also, according to the discussion following \eqref{hat P tau}, one should expect that the price of a real-world dark pool should take at least some fraction of temporary price impact into account,  i.e., one should expect $\kappa>0$ for real-world dark pools.

\item[(iii)]Condition \eqref{Delta>0 for alpha=1 Condition} means that  the execution of a dark-pool order of size $Y_i$ needs to generate \lq\lq slippage" of at least $\frac\gamma2 Y_i^2$. This latter amount is just equal to the costs from permanent impact one would have incurred by executing the order at the exchange.  With this amount of slippage,   savings by executing an order not at the exchange but at a dark pool    only arise from savings on temporary impact but not on permanent impact. It seems that dark pools that are currently operative do not charge transaction costs or taxes of this magnitude.  Nevertheless, our theorem states that a penalization of size $\beta(Y_i)Y_i\ge\frac\gamma2 Y_i^2$ is needed for a robust stabilization of  the model against irregularities. 
\end{enumerate}
\end{remark}

\begin{remark} The implication (b)$\Rightarrow$\ref{thm41point-c} in Theorem \ref{model-parameters} is proved by constructing price manipulation strategies for the cases in which some of the conditions in \ref{thm41point-c} are not satisfied.  Most of these price manipulation strategies involve only a  small order $\HX$ placed in the dark pool and another small order stream $(\xi_t)$ used for manipulating  the exchange-quoted price until an execution in the dark pool occurs or a certain deadline $\rho$ passes. Afterwards, the remaining inventory is liquidated at the exchange over a time period $[\rho,T]$ that is sufficiently large for the resulting temporary price impact to be small.   Since order sizes are small, so will be the inventory remaining  at time $\rho$.  So, even if we formally send $T$ to infinity in the proofs, choosing $T$ in the range of seconds or minutes may be sufficient for similar strategies to work in practice. The price manipulation we use in our proofs may actually have some similarity with some strategies used by high-frequency arbitrageurs.  The only exception of the rule of small order sizes occurs when proving \eqref{Delta>0 for alpha=1 Condition} for large $|y|$, because then we will need to place a dark-pool order   $\HX=y$. Nevertheless, it is easily  possible to require an upper limit $\widehat y>0$ on the admissible size $|\HX|$ of dark pool orders and obtain a variant of Theorem  \ref{model-parameters} in which \eqref{Delta>0 for alpha=1 Condition} is required only for $|y|\le\widehat y$.
\end{remark}

 Theorem \ref{model-parameters} gives a complete characterization of the regularity of a dark-pool extension when the underlying Almgren--Chriss model $(\gamma, h, P^0)$ and the possible time horizon may change. In fact, we only need to alter the parameters $h$ and $T$. 
 If all parameters are fixed, the situation becomes more involved. The following two propositions give some implications that hold in the generality of Assumption \ref{DarkPoolAssumption1}. The first general proposition deals with the \lq\lq slippage parameter" $\beta(\cdot)$.

\begin{proposition}\label{GivenACmodelAndTProp1}Suppose an Almgren--Chriss model with parameters $(\gamma, h, P^0)$ and its dark-pool extension  $(\alpha,\beta,\kappa,(\tau_i),(\widetilde Y_i))$ have been fixed. When there is no price manipulation for a given time horizon $T>0$, then
\begin{equation}\label{beta inequality finite T}
\beta(y)y\ge\gamma\Big(\alpha-\frac12\Big)y^2+yh(-y/T)\qquad\text{for all $y$.}
\end{equation}
In particular, we must have that 
\begin{equation}\label{beta inequality all T}
|\beta(y)|\ge\gamma\Big(\alpha-\frac12\Big)|y|\qquad\text{for all $y$}
\end{equation}
when there is no price manipulation for all $T>0$.
\end{proposition}

The next proposition deals with the parameters $\alpha$ and $\kappa$ under the additional assumption that there is equality in \eqref{beta inequality all T}. This equality implies in particular that $\alpha\ge1/2$ and that \lq\lq slippage" has the minimal value that is necessary to exclude price manipulation.

\begin{proposition}\label{GivenACmodelAndTProp2}Suppose that an Almgren--Chriss model with parameters $(\gamma, h, P^0)$ and its dark-pool extension  $(\alpha,\beta,\kappa,(\tau_i),(\widetilde Y_i))$ have been fixed.  Suppose moreover that there is equality in \eqref{beta inequality all T} and that there is no price manipulation for all $T>0$.
Then 
$\alpha=1$ and $\kappa =0$.
\end{proposition}

In the next section, we will analyze the regularity of a  concrete class of dark-pool extensions of a fixed Almgren--Chriss model. We will see that the question of regularity exhibits a rich mathematical structure and that the absence of price manipulation can no longer be characterized solely in terms of $\alpha$, $\beta$, and $\kappa$.  For instance, it will follow from Corollary \ref{CounterExCor} that even in the case $\alpha=\beta=0$ it may happen that there is no price manipulation for all $T>0$, but this situation is then characterized in terms of relations between $\gamma$, $h$, and the law of $\tau_1$.

\subsection{Regularity and irregularity for special model characteristics}\label{SpecResSec}

In this section, we will investigate in more detail the regularity and irregularity of a dark-pool extension of a \emph{fixed} Almgren--Chriss model. To this end, we will assume throughout this section that slippage is zero, $
\beta=0$, 
which is the natural (naive) first guess in setting up a dark-pool model. We know from Theorem \ref{model-parameters}, though, that there must be some Almgren--Chriss model such that there is price manipulation for sufficiently large  time horizon $T$.

 First, we will look into the role played by  $T$ in the existence of price manipulation. We  will argue that there exists a critical threshold $T^*\ge0$ such 
that there is no price manipulation for $T< T^*$ but price manipulation does exist for $T>T^*$. We will show that all three cases $T^*=\infty$, $0<T^*<\infty$, and $T^*=0$ can occur. Second, we will analyze the stronger requirements of absence of transaction-triggered price manipulation and of nonnegative expected liquidation costs. We will find situations in which there is no price manipulation for all $T>0$ but where the condition of nonnegative expected liquidation costs fails and where there is transaction-triggered price manipulation for sufficiently large $T$.

The results in this section serve a two-fold purpose. First, they illustrate that the question of regularity for a dark-pool extension of an Almgren--Chriss model can become quite subtle when condition \ref{thm41point-c} from Theorem \ref{model-parameters} is not satisfied. Second, this section also provides a case study for the regularity of market impact models in general. 

 We will make the following simple but natural assumption on the dark-pool extension defined through $(\alpha,\beta,g,(\tau_i),(\widetilde Y_i))$. Note that part (b) of the following assumption implies all parts of Assumption \ref{DarkPoolAssumption1}. 

\begin{assumption}\label{IndependenceAssumption}We assume the following conditions throughout Section \ref{SpecResSec}.
\begin{enumerate}
\item Slippage is zero: $
\beta=0$.
\item The process  $(N_t)$, as defined in \eqref{NtDefEq}, is a standard Poisson process with parameter $\theta>0$ and   $(\widetilde Y_i)$ are i.i.d. random variables with common distribution $\mu$ on $(0,\infty]$ such that $\mu([x,\infty])>0$ for all $x>0$. We also assume that the stochastic processes $(P^0_t)$, $(N_t)$, and $(\widetilde Y_i)$ are independent. 
\end{enumerate}
\end{assumption}

In Assumption \ref{IndependenceAssumption} (b) we do not exclude the possibility that $\widetilde Y_i$ takes the value $+\infty$ with positive probability. The particular case $\widetilde Y_i=+\infty$ $\mathbb P$-a.s., corresponding to $\mu=\delta_\infty$, can be regarded as a convenient description of the limiting case of infinite liquidity of incoming dark-pool orders. 

\bigskip

In Propositions \ref{thm-no-opt-strategies} and \ref{infinite-DP-thm-small-T}, we will consider the situation in which  $\alpha=1$.  In view of our assumption $\beta=0$, Proposition \ref{GivenACmodelAndTProp1} implies that there will be price manipulation for sufficiently large $T$. The following proposition shows in particular that one then can also generate arbitrarily negative expected costs. In contrast to the situation in many other market impact models, this conclusion is not  obvious in our case, because \blue{our} expected costs are typically not a convex functional of admissible strategies. 

\begin{proposition} \label{thm-no-opt-strategies}
Suppose that an Almgren--Chriss model  has been fixed and that $\alpha = 1$.  Then, for any $X_0 \in \mathbb R$,
\[ \lim_{T \uparrow \infty} \inf_{\chi \in \mathcal X(X_0,T)} \mathbb E[\,\mathcal C_T^\chi\,] = - \infty.\]
In particular, the condition of nonnegative expected liquidation costs is violated.
\end{proposition}

Now we examine in more detail the role played by $T$ in the existence of price manipulation. First, we show that for a certain class of models there is no price manipulation for small $T$.

\begin{proposition} \label{thm-noPMforsmallT}
Let $\kappa = 0$ and $h(x) = \eta x$. If $T\le \frac{2\eta}{\gamma}$, then there is no price manipulation.  If moreover $\alpha = 1$, then there is no price manipulation if and only if $T\le \frac{2\eta}{\gamma}$.
\end{proposition}

Since the class $\mathcal X(X_0,T)$ of admissible strategies increases with $T$, the existence of price manipulation for some $T$ implies the existence of price manipulation for any $T'\ge T$. Hence there exists a  critical value $T^*$ such that there is no price manipulation for $T < T^*$  but price manipulation does exist for $T>T^*$. For $\alpha=1$ and linear temporary impact, $h(x)=\eta x$, it follows from Proposition \ref{thm-noPMforsmallT} that $T^*=\frac{2\eta}{\gamma}$. The next proposition  shows that $T^* = 0$ for  $\alpha=1$ and  temporary impact with sublinear growth.

\begin{proposition} \label{infinite-DP-thm-small-T} Suppose that an Almgren--Chriss model  has been fixed and that $\alpha = 1$.
If  $h$ has sublinear growth, i.e., 
\[ \lim_{|x|\rightarrow \infty} \frac{h(x)}{x}  = 0,\]
then there is price manipulation for every $T>0$. 
\end{proposition}

\commentOut{Report 2, point 3: Referee claims that presentation is confusing, so we should make more clear that we consider first alpha=1, then alpha=0. Probably subsubsections/more subsections?}After   having considered the case $\alpha =1$, we will assume in the sequel that
\begin{equation}\label{ZeroAssumptionEq}
\alpha=0, \qquad \kappa=0, \qquad\text{and}\qquad h(x)=\eta x,
\end{equation}
in addition to Assumption \ref{IndependenceAssumption}.  Note that we assume $\kappa=0$ since this is one of the conditions required by Theorem \ref{model-parameters} for regularity. Taking $\kappa>0$ will typically make the model even less regular. The assumption $\alpha=0$ can again be regarded as the natural (naive) first guess in setting up a dark-pool model. But in view of Theorem~\ref{model-parameters} the assumption $\alpha=0$ implies the existence of an Almgren--Chriss model  for which there is price manipulation and for which the condition of nonnegative expected liquidation costs is violated for sufficiently large $T$. Our aim is to  give a refined analysis for the specific class of Almgren--Chriss models with linear temporary price impact, $h(x)=\eta x$, which is probably the most commonly studied case in the academic market impact literature.
We first take a  look at the condition of nonnegative expected liquidation costs. 

\begin{proposition} \label{thm-sunny-no-pm} 
 Consider a fixed Almgren--Chriss model  and suppose that condition \eqref{ZeroAssumptionEq} holds.
\begin{enumerate}
\item If $\frac{\gamma}{\eta }<2\theta$, we have for any $X_0\in\mathbb R\setminus\{0\}$,
\begin{equation}  \label{thm-sunny-no-pm-eq}
\lim_{T \uparrow \infty} \inf_{\chi \in \mathcal X(X_0,T)} \mathbb E[\,\mathcal C_T^\chi\,] \le - \frac{\gamma^2}{2}  X_0^2 \frac{1}{2 \eta \theta - \gamma} < 0.
\end{equation}
\item If either $\frac{\gamma}{\eta } = 2\theta$ and $X_0 \ne 0$ or $\frac{\gamma}{\eta } > 2\theta $, then 
\[ \lim_{T \uparrow \infty} \inf_{\chi \in \mathcal X(X_0,T)} \mathbb E[\,\mathcal C_T^\chi\,] = - \infty.\]
\end{enumerate}
\end{proposition}

In particular, the condition of nonnegative expected liquidation costs is violated in both parts of Proposition \ref{thm-sunny-no-pm}. Proposition \ref{hierarchy-lemma} therefore yields that there is also transaction-triggered price manipulation. In fact, the next proposition shows that price manipulation and negative expected liquidation costs can only be realized by using transaction-triggered price manipulation. Therefore,  only a strategy that manipulates the exchange-quoted price can be more profitable than other strategies.

\begin{proposition} \label{trans-trigg-lemma} Consider a fixed Almgren--Chriss model  and suppose that condition \eqref{ZeroAssumptionEq} holds.  If $X_0 \ge 0$ and $\xi_t \le 0$ for all $t$ (or $X_0 \le 0$ and $\xi_t \ge 0$ for all $t$), then $\mathbb E[\,\mathcal C_T\,] \ge 0$, i.e., the violation of nonnegative expected liquidation costs in Proposition \ref{thm-sunny-no-pm} can only be achieved by intermediate buy (sell) trades at the exchange during an overall sell (buy) program.
\end{proposition}

Some of the preceding results can be strengthened in the  limiting case $\mu=\delta_\infty$, which corresponds to infinite liquidity of incoming dark-pool orders.  
We first show that then \eqref{thm-sunny-no-pm-eq} becomes an equality.

\begin{proposition} \label{infinite-DP-thm-sunny-no-pm} 
Consider a fixed Almgren--Chriss model. Suppose moreover that condition \eqref{ZeroAssumptionEq} holds and that $\mu=\delta_\infty$. 
Then,  for $X_0 \in \mathbb R$ and $\frac{\gamma}{\eta} < 2 \theta$,
\begin{equation}\label{infinite-DP-thm-sunny-no-eq} 
\lim_{T \uparrow \infty} \inf_{\chi \in \mathcal X(X_0,T)} \mathbb E[\,\mathcal C_T^\chi\,] = - \frac{\gamma^2 }{2} X_0^2 \frac{1}{2 \eta \theta - \gamma}.
\end{equation}
\end{proposition}

Equation \eqref{infinite-DP-thm-sunny-no-eq}  is remarkable, because it implies on the one hand that the condition of nonnegative expected liquidation costs is violated for all $X_0\neq0$. By taking $X_0=0$ we see, on the other hand, that there is no price manipulation for all $T>0$ and so $T^*=\infty$. We actually have the following result:

\begin{corollary}\label{CounterExCor} Consider a fixed Almgren--Chriss model. Suppose moreover that condition \eqref{ZeroAssumptionEq} holds and that $\mu=\delta_\infty$. 
Then there is no price manipulation for every $T>0$ if and only if $\frac{\gamma}{\eta }\le 2\theta$.
\end{corollary}

By comparing the preceding result with Propositions \ref{hierarchy-lemma} and \ref{infinite-DP-thm-sunny-no-pm},  we 
arrive \commentOut{Report 2, minor comment 8}at the following statement.

\begin{corollary}Under the assumptions of Corollary \ref{CounterExCor} there is {always} transaction-triggered price manipulation for sufficiently large $T$. Standard price manipulation, however, exists only for $\frac{\gamma}{\eta }> 2\theta$.
\end{corollary}

\subsection{Optimal trade execution strategies}
\label{section_optliq}

In this section, we illustrate some of our results by determining an optimal strategy for selling $X_0>0$ shares. To this end, we will make a number of simplifying assumptions, because our main goal is to analyze the regularity of the model. In particular,  our optimality criterion is the minimization of expected costs; we do not consider risk aversion. 

We fix an Almgren--Chriss model  and assume that Assumption \ref{IndependenceAssumption} (b) holds and that
\begin{equation}\label{LiquidationAssEq}
 \alpha = 1, \; \beta(y) = \frac{\gamma}{2}y, \; \kappa = 0.
\end{equation}
Then Theorem~\ref{model-parameters} guarantees that there is no price manipulation. For simplicity, we will also assume that all matching dark-pool orders are large compared to our's in the sense that
\begin{equation}\label{LiquidationAss2Eq}
\mu=\delta_\infty.
\end{equation} Then the entirety of the dark-pool order $\widehat X $ will either be filled  when $\tau_1\le\rho$ or 
it will be cancelled when $\tau_1>\rho$. 
In this setting, an admissible strategy $\chi=(\widehat X,\Xi,\xi,\rho)\in\mathcal{X}(X_0,T)$ will be called a \emph{single-update strategy} if $\rho$ is a deterministic time in $[0,T)$ and $\xi$  is predictable with respect to the filtration generated by the stochastic process $\mathbbmss 1_{ \{ \tau_1 \le t \} }$, $t\ge0$. Due to \eqref{LiquidationAss2Eq}, the value of $\Xi$ is irrelevant.

Note that the process $\xi$ of a single-update strategy evolves deterministically until there is an execution in the dark pool, i.e., until time $\tau_1$. At that time,  $\xi$ can be updated. But the update will only depend on the time $\tau_1$ and not on any other random quantities.
In particular, $\xi$ can be written as
\begin{equation} 
\xi_t = \begin{cases}
\xi^0_t, & \textrm{ if } t \le \tau_1 \textrm { or } \tau_1>\rho,\\
\xi^1_t, & \textrm{ if } t > \tau_1 \textrm{ and } \tau_1\le\rho,
\end{cases}
\end{equation}
where $\xi^0$ is deterministic and $\xi^1$ depends on $\tau_1$. 

\begin{proposition} \label{thm-semi-deterministic}Suppose that Assumption \ref{IndependenceAssumption} {\rm (b)}, \eqref{LiquidationAssEq}, and \eqref{LiquidationAss2Eq} hold and that $|h(x)|\to\infty$ as $|x|\to\infty$. 
For any $X_0\in\mathbb R$  and $T>0$ there exists a single-update strategy  that minimizes the expected costs $ \mathbb E [ \,\mathcal C^\chi_T\,]$ in the class of all admissible strategies $\chi\in\mathcal X(X_0,T)$. Moreover, the part $\xi^1$ of any optimal single-update strategy must satisfy 
\begin{equation}\label{optimal xi 1 eq}
\xi^1_t=\begin{cases}\displaystyle\frac{-X_{\tau_1}-\widehat X}{T-{\tau_1}}&\text{on }\{ {\tau_1} \le\rho \},\\ &\\
\displaystyle\frac{-X_\rho}{T-\rho}&\text{on }  \{ \rho<{\tau_1} \};
\end{cases}
\end{equation}
\end{proposition}

Now we show how  the components $\xi^0$, $\HX$, and $\rho$ of  an optimal single-update strategy can be computed. To this end, we make the additional simplifying assumption that temporary impact is linear, $h(x)=\eta x$. 
It  follows from Equation \eqref{SingleUpdateRevenuesEq} in the proof of Proposition \ref{thm-semi-deterministic} that the expected costs of a single-update strategy $\chi\in\mathcal X(X_0,T)$  satisfying \eqref{optimal xi 1 eq} are given by
\begin{equation}\label{ExpectedRevenuesOptLiqEq}
 \mathbb E[\,\mathcal C^\chi_T\,] 
  = \frac 12 \gamma X_0^2 +\int_0^\rho \eta (\xi^0_s)^2 e^{-\theta s}\, ds+\eta e^{-\theta\rho}\frac{(X_0+\int_0^\rho \xi^0_s\, ds)^2}{T-\rho}+\int_0^\rho \eta\theta e^{-\theta t} \frac{(X_0+\int_0^t \xi^0_s \, ds +\HX)^2}{T-t}\, dt
\end{equation}
A standard calculation shows that  $X^0_t := X_0 +\int_0^t \xi^0_s \, ds$, $0\le t \le \rho $, minimizing this expression is the solution of the  Euler--Lagrange equation
\begin{equation*}  
	- \ddot X^0_t + \theta \dot X^0_t + \theta \frac{X^0_t + \HX}{T - t} = 0
\end{equation*}
with initial condition $X_0^0=X_0$ and a terminal condition $X^0_\rho$ that will be determined later. By using the computer algebra software Mathematica, we found the analytic solution
\begin{eqnarray*} \lefteqn{X^0_t=\red{-}\widehat X+}\\
&& \Big(\theta  T e^{\theta  T} (T-\rho ) (\text{Ei}(-T \theta
   )-\text{Ei}(\theta  (\rho -T)))-\rho +T (1-e^{\theta  \rho })\Big)^{-1}\bigg\{-e^{\theta  t}\rho  X_0\\
&&+(t-T) \big(X_0 e^{\theta  \rho }-X^0_\rho-\widehat X\big)+\theta  (T-t) e^{\theta  T} \Big[\text{Ei}((t-T) \theta ) (T (X_0-X^0_\rho-\widehat X)\\
&&-\rho 
   X_0)+X_0 (\rho -T) \text{Ei}(\theta  (\rho -T))+T (X_\rho^0+\widehat X) \text{Ei}(-T \theta
   )\Big]+e^{\theta  t} T
   (X_0-X^0_\rho-\widehat X)\bigg\},
\end{eqnarray*}
where 
$\text{Ei}(t) = \int_{-\infty}^t s^{-1}{e^s} \, ds$ is the exponential integral function. The constants $\rho$, $\widehat X$ and $X^0_\rho$ can then be determined by optimizing the expression \eqref{ExpectedRevenuesOptLiqEq} numerically.

\section{Conclusion} \label{ConclusionSection}

We have analyzed the regularity of a class of dark-pool extensions of an Almgren--Chriss model and found that such models admit price manipulation strategies unless the model parameters satisfy certain restrictions. The corresponding parameter values will typically differ strongly from  values found in empirical analysis or calibration of real-world dark pools.  Our results can therefore provide some indication that dark pools may create market inefficiencies and disturb the price finding mechanism of markets. 

In concrete realizations of our model, we have furthermore provided a comparative analysis of various regularity notions for market impact models. We  have found that the validity of these notions can depend in a subtle way  on the interplay of all  model parameters and on the liquidation time constraint.

\section{Proofs}\label{ProofsSection}

Recall from Assumption \ref{DarkPoolAssumption1} that the martingale property of $(P^0_t)$ is retained by passing to the enlarged filtration $(\mathcal G_t)$. Next,  for an admissible strategy $\chi=(\HX,\M, \xi, \rho)$, the asset position process  $X$ defined in \eqref{StrategyEq} is an admissible integrand for $P^0$ since it is bounded and $(\mathcal G_t)$-predictable, because $X$ is both  left-continuous and $(\mathcal G_t)$-adapted. Recall also that $ f(x) = x \, \TI(x)$.

\begin{lemma} \label{revenuesT} The costs of an admissible strategy $\chi=(\HX,\M, \xi, \rho)$  for given $X_0$ and $T$ are given by
\begin{eqnarray*} 
\mathcal C^\chi_T & = & -\int_0^TX_t\,dP^0_t+\frac\gamma2\bigg(X_0+\sum_{i=1}^{N_\rho} Y_i\bigg)^2+\int_0^T f(\xi_t)\,dt\\
&&+\sum_{i=1}^{N_{\rho}}Y_i\bigg(\gamma\int_0^{\tau_i}\xi_s\,ds-\gamma\alpha X_{\tau_i+}+\kappa h(\xi_{\tau_i})+\beta(Y_i)\bigg).
\end{eqnarray*}
\end{lemma}

\begin{proof}
First we prove that
\begin{equation}\label{lemma revenuesT first eq}
 \int_0^T \xi_t P^0_t \, dt + \sum_{i=1}^{N_\rho} Y_i P^0_{\tau_i} = -X_0 P^0_0 - \int_0^T X_t \, dP^0_t.
\end{equation}
To this end, we use first \eqref{liquidationconstraint} and integration by parts to get
\begin{eqnarray*}
X_0 P^0_0 + \int_0^T X_t \, dP^0_t & = & X_0 P^0_0 + \int_0^T \left(X_0 + \int_0^t \xi_s \, ds + Z^\rho_{t-}\right) \, dP^0_t \\
& = & X_0 P^0_0 + X_0 (P^0_T - P^0_0) + \int_0^T \int_0^t \xi_s \, ds\, dP^0_t + \int_0^T Z^\rho_{t-} \, dP^0_t\\
& = & - P^0_T Z^\rho_{T-}   - \int_0^T \xi_t P^0_t \, dt  + Z^\rho_TP^0_T-\int_0^TP_{t-}^0\,dZ^\rho_t -[P^0,Z]_t.\end{eqnarray*}
Since $\rho<T$ the two terms $- P^0_T Z^\rho_{T-} $ and $Z^\rho_TP^0_T-$ cancel each other out. Moreover, the
 definition \eqref{Zt def eq} of $Z$ implies that 
$$\int_0^TP_{t-}^0\,dZ^\rho_t  +[P^0,Z]_t=\sum_{i=1}^{N_\rho}  P^0_{\tau_i-}Y_i+ \sum_{i=1}^{N_\rho}  \Delta P^0_{\tau_i}Y_i=\sum_{i=1}^{N_\rho}  P^0_{\tau_i}Y_i.
$$
Putting everything together yields \eqref{lemma revenuesT first eq}.

Therefore, $\mathbb{P}$-a.s.,
\begin{eqnarray*}\cC_T^\chi&=&X_0P_0^0+\int_0^T\xi_tP^\chi_t\,dt+\alpha\gamma\int_0^T Z^\rho_{t}\xi_t\,dt+\sum_{i=1}^{N_{\rho}}Y_i\Big(\widehat P^\chi_{\tau_i}+\alpha\gamma Z_{\tau_i-}+\beta(Y_i)\Big)\\
&=&X_0P_0^0+\int_0^T\xi_t\Big(P^0_t+\gamma\int_0^t\xi_s\,ds+h(\xi_t)\Big)\,dt+\alpha\gamma\int_0^T Z^\rho_{t}\xi_t\,dt\\
&&+\sum_{i=1}^{N_{\rho}}Y_i\Big( P^0_{\tau_i}+\gamma\int_0^{\tau_i}\xi_s\,ds+\kappa h(\xi_{\tau_i})+\alpha\gamma \sum_{j=1}^{i-1}Y_j+\beta(Y_i)\Big)\\
&=&-\int_0^TX_t\,dP^0_t+\gamma \int_0^T \int_0^t \xi_s\,ds \, \xi_t\,dt+\int_0^T f(\xi_t)\,dt+\gamma \alpha \sum_{i=1}^{N_\rho} Y_i \int_{\tau_i}^T \xi_t\,dt\\
&&+\sum_{i=1}^{N_{\rho}}Y_i\Big( \gamma\int_0^{\tau_i}\xi_s\,ds+\kappa h(\xi_{\tau_i})+\alpha\gamma\sum_{j=i+1}^{N_\rho} Y_j+\beta(Y_i)\Big)\\
&=&-\int_0^TX_t\,dP^0_t+\frac\gamma2\bigg(\int_0^T\xi_t\,dt\bigg)^2+\int_0^T f(\xi_t)\,dt\\
&&+\sum_{i=1}^{N_{\rho}}Y_i\bigg(\gamma\int_0^{\tau_i}\xi_s\,ds+\gamma\alpha\Big(\int_{\tau_i}^T \xi_t\,dt+\sum_{j=i+1}^{N_\rho} Y_j \Big)+\kappa h(\xi_{\tau_i})+\beta(Y_i)\bigg)\\
&=&-\int_0^TX_t\,dP^0_t+\frac\gamma2\bigg(X_0+\sum_{i=1}^{N_\rho} Y_i\bigg)^2+\int_0^T f(\xi_t)\,dt\\
&&+\sum_{i=1}^{N_{\rho}}Y_i\bigg(\gamma\int_0^{\tau_i}\xi_s\,ds-\gamma\alpha X_{\tau_i+}+\kappa h(\xi_{\tau_i})+\beta(Y_i)\bigg).
\end{eqnarray*}
In the last step, we have again used the fact that $X_T=X_{T+}=0$ $\mathbb{P}$-a.s.
\end{proof}

\begin{proof}[Proof of Proposition~\ref{hierarchy-lemma}]
(a): Assume $X_0 \ge 0$, and let the trading strategy be sell-only, i.e. $\xi_t \le 0$ for all $t$ and $Y_i \le 0$ for all $i$. Then in particular $P^\chi_t \le P^0_t$ for all $t$ and $\widehat P^\chi_{\tau_i} \le P^0_{\tau_i}$ for all $i$. From \eqref{lemma revenuesT first eq} we get that 
$$\mathcal C_T \ge X_0 P^0_0+ \int_0^T \xi_s P^0_s \, ds + \sum_{i=1}^{N_{T\wedge \rho}} Y_i P^0_{\tau_i} = - \int_0^T X_t \, dP^0_t.$$
 Since $(P^0)$ is a  $(\mathcal{G}_t)$-martingale by Assumption \ref{DarkPoolAssumption1} (c) and $X_t$ is bounded according to Definition \ref{strategy chi def}, we have $\mathbb E[\,\mathcal C_T\,] \ge 0$ for such a trading strategy. Absence of transaction-triggered price manipulation implies that the expected costs cannot be decreased by intermediate sell trades. Therefore  $\mathbb E[\,\mathcal C_T^\chi \,] \ge 0$ for all admissible trading strategies $\chi$. The case $X_0 \le 0$ works analogously.

(b): By setting $X_0 = 0$ in (\ref{illiquid-cond}) we find that $\mathbb E[\,\mathcal C^\chi_T\,] \ge 0$ for round trips $\chi$.
\end{proof}

In the following, we will consider round trips $\chi=(\HX,\M, \xi, \rho)$ that cancel the order in the dark pool after the first execution, i.e. $X_0 = 0$ and there exists some $r<T$ such that
\begin{equation}\label{rho in proof of Thm 4.1}
\rho:=\tau_1\indf{\tau_1\le r\text{ and }\widetilde Y_1\ge\Xi}+r\indf{\tau_1>r\text{ or }\widetilde Y_1<\Xi}.
\end{equation}
We will also take
\begin{equation}\label{Xi=HX}
\Xi=|\HX|
\end{equation}
and the following class of strategies $\xi$, which  depend only on $\tau_1$ and $\widetilde Y_1$:
\begin{equation}\label{xi in proof of Thm 4.1}
\xi_t=-x\ind{[0,r]}(t)+\frac{rx }{T-r}\ind{(r,T]}(t)\cdot\ind{\{\text{$\tau_1>r$ or $\widetilde Y_1<\Xi$}\}}+\frac{rx-\HX}{T-r}\ind{(r,T]}(t)\cdot\ind{\{\text{$\tau_1\le r$ and $\widetilde Y_1\ge\Xi$}\}}.
\end{equation}
 With Lemma~\ref{revenuesT} we find that the costs of such a \emph{single-update round trip} $\chi=(\HX,\M, \xi, \rho)$ are 
\begin{eqnarray*}\mathcal C^\chi_T &=&- \int_0^T X_t \, dP^0_t +rf(-x)+(T-r)f\Big(\frac{rx}{T-r}\Big)\ind{A^c}+(T-r)f\Big(\frac{rx-\HX}{T-r}\Big)\ind A\\
&&\quad+\ind A \HX\Big(\frac{\gamma}{2} \HX+\gamma X_{\tau_1 -} - \gamma\alpha (X_{\tau_1-}+Y_1)+\kappa h(\xi_{\tau_i})+\beta(Y_1)\Big),
\end{eqnarray*}
where $A=\{\tau_1\le r\text{ and }\widetilde Y_1\ge\HX\}$.
Consequently, 
the expected costs of $\chi$ are
\begin{eqnarray}
\mathbb E[\,\cC^\chi_T\,]&=&r f(-x)+(T-r)f\Big(\frac{rx}{T-r}\Big)\mathbb P[\,\tau_1>r\text{ or }\widetilde Y_1<|\HX|\,]\label{costs of xi in proof of Thm 4.1}\\
&&+\mathbb E\bigg[\,(T-r)f\Big(\frac{rx-\HX}{T-r}\Big)+\phi(\HX)-\gamma(1-\alpha)\tau_1x\HX+\kappa h(-x)\HX;\ \tau_1\le r\text{ and }\widetilde Y_1\ge|\HX|\,\bigg],\nonumber
\end{eqnarray}
where 
\begin{equation}\label{Delta}
\OldDelta(y):=\frac{\gamma}{2}y^2-\alpha \gamma y^2+\beta(y)y.
\end{equation}

\begin{proof}[Proof of Proposition \ref{GivenACmodelAndTProp1}] We first prove the inequality \eqref{beta inequality finite T}, which can also be written as $\phi(y)\ge yh(-y/T)$ for $\phi$ is as in \eqref{Delta}. Suppose by way of contradiction that there exists some $y$ such that that $\phi(y)< yh(-y/T)$ but that there is no price manipulation.  Consider the single-update round trip with $\widehat X=y$, minimum quantity level $\M=|\widehat X|$, cancellation time $r \in(0,T)$, and $x=0$ in \eqref{xi in proof of Thm 4.1}.
By \eqref{costs of xi in proof of Thm 4.1}, the expected costs of this strategy are
\begin{eqnarray} 
\mathbb E[\,\mathcal C_T\,] =\Big(\phi(y)-yh\Big(\frac{-y}{T-r}\Big)\Big)\mathbb P\left[\,\tau_1\le r\text{ and }\widetilde Y_1\ge \Xi\,  \right].\nonumber
\end{eqnarray}
When $r$ is sufficiently small, our assumption $\phi(y)< yh(-y/T)$ implies that the right-hand side will be strictly negative, and so there will be price manipulation. But this is the desired contradiction. Inequality \eqref{beta inequality all T} now follows by sending $T$ to infinity and using the continuity of $h(x)$ at $x=0$.
\end{proof}

\begin{proof}[Proof of  Proposition \ref{GivenACmodelAndTProp2}]
To prove $\kappa =0$, we assume by way of contradiction that $\kappa>0$. As before, we take a single-update round trip with \eqref{Xi=HX}, \eqref{xi in proof of Thm 4.1}, and $\HX,\,x\neq 0$ such that both $\HX$ and $x$ have the same sign and $|\HX|\le1$. Since there is equality in \eqref{beta inequality all T}, we have $\phi(y)=0$ for all $y$. When sending $T$ to infinity in \eqref{costs of xi in proof of Thm 4.1}, we thus get with dominated convergence that
\begin{eqnarray*}
\lim_{T\uparrow\infty}\mathbb E[\,\cC_T\,]&=&r f(-x)+\mathbb E\bigg[\,-\gamma(1-\alpha)x\tau_1\HX+\kappa h(-x)\HX;\ \tau_1\le r\text{ and }\widetilde Y_1\ge|\HX|\,\bigg]\nonumber\\
&\le&-h(-x)\big(rx-\kappa\HX r\lambda_0\lambda_1(1)\big),\label{kappa=0 cost estimate 1}
\end{eqnarray*}
which is strictly negative for arbitrary $\HX>0$ and sufficiently small $x>0$. Hence there is price manipulation for sufficiently large $T$. This proves $\kappa =0$.

Now we prove that we must have $\alpha=1$. We assume $\alpha<1$ by way of contradiction and let $r=1$, $x>0$, $\HX\in(0,1]$ in \eqref{costs of xi in proof of Thm 4.1}. We also use the already established fact that $\kappa =0$. Sending $T$ to infinity yields that 
\begin{eqnarray*}
\lim_{T\uparrow\infty}\mathbb E[\,\cC_T\,]&=&r f(-x)+\mathbb E\bigg[\,-\gamma(1-\alpha)x\tau_1\HX;\ \tau_1\le r\text{ and }\widetilde Y_1\ge|\HX|\,\bigg]\\
&\le&-x\Big(r h(-x)+\gamma(1-\alpha) \HX\,\mathbb E[\,\tau_1;\ \tau_1\le r\text{ and }\widetilde Y_1\ge1\,]\Big).
\end{eqnarray*}
 Again, the latter expression is strictly negative for arbitrary $\HX$ and sufficiently small $x$. This proves that there is price manipulation for sufficiently large $T$ and in turn establishes $\alpha=1$. \end{proof}

The function $f(x)=xh(x)$ is convex according to Assumption \ref{AlmgrenChrissAssumption} and Definition \ref{temporary impact def}. Hence, $f$ admits left- and right-hand derivatives, $f'_-(x)$ and $f'_+(x)$, which are nondecreasing functions of $x$.

\begin{lemma}\label{stimmt das Lemma} Let $h$ be an impact function and $f(x):=xh(x)$. Then $f(x)$ is  differentiable at $x=0$ with $f'(0)=0$. Moreover, $f'_\pm(x)\to0$ as $x\to0$.
\end{lemma}

\begin{proof} First note that we have 
\begin{equation}\label{kjhgvk}
\frac{f(x)-f(0)}{x-0}=\frac{f(x)}{x}=h(x)\longrightarrow 0\qquad\text{as $x\to0$.}\end{equation}
Hence, $f$ is differentiable at $x=0$ with $f'(0)=0$. Moreover it is well known that for a convex function $f$ we have $f'_\pm(x)\to f'(x_0)$ when $f$ is differentiable at $x_0$. \end{proof}

\begin{proof}[Proof of Theorem \ref{model-parameters}]
 The implication (a)$\Rightarrow$(b) follows immediately by taking $X_0=0$. 

(b)$\Rightarrow$(c): We already know from Proposition \ref{GivenACmodelAndTProp1} that we must have \eqref{beta inequality all T}, which yields \eqref{Delta>0 for alpha=1 Condition} when $\alpha=1$. Thus, it  remains to show that $\kappa=0$ and $\alpha=1$. 

We start by showing that $\kappa=0$.  To this end, we assume by way of contradiction that  there is no price manipulation but $\kappa>0$. When the function $\phi$ defined in \eqref{Delta} is identically zero, this result follows from Proposition \ref{GivenACmodelAndTProp2}. Otherwise, we will now construct a suitable impact function $h$ so that there is price manipulation. To this end, we can assume without loss of generality that $\phi$ is strictly convex---if not, we can replace $\beta$ by another impact function $\widetilde\beta$ with $|\widetilde\beta|\ge|\beta|$ for which $\widetilde\phi(y):=\gamma(\frac12-\alpha)y^2+y\widetilde\beta(y)$ is strictly convex and satisfies $\widetilde\phi\ge\phi\ge0$. 
By the same argument, we may assume without loss of generality that $\beta$, and hence also $\phi$, are differentiable  on $\mathbb R\setminus\{0\}$ and that $\varphi(y)=\gamma(\frac12-\alpha)y+\beta(y)$ is an impact function in the sense of Definition \ref{temporary impact def}. Now we define
$$ f(x):=\int_0^{|x|}\sqrt{\phi'({\sqrt{y}})}\,dy.
$$
This definition makes sense, because the derivative $\phi'$ of the  convex function $\phi:\mathbb R\to[0,\infty)$ is bounded on every compact interval.
Then $f$ has a strictly increasing derivative and hence is strictly convex. Also, $f(x)\ge0$ with equality if and only if $x=0$. Since $f$ is strictly convex,
$$h(x):=\frac{f(x)}{x}=\frac{f(x)-f(0)}{x-0},\qquad x\neq 0
$$
is strictly increasing in $x$ and continuous on $\mathbb R\setminus\{0\}$. Moreover, it satisfies
$h(0):=\lim_{x\to 0}h(x)=0$, because 
$$|h(x)|=\frac1{|x|}\int_0^{|x|}\sqrt{\phi'({\sqrt{y}})}\,dy\le\sup_{0\le y\le|x|}\sqrt{\phi'({\sqrt{y}})}=\sqrt{\phi'({\sqrt{|x|}})}\longrightarrow0
$$
 as $x\to0$ by Lemma \ref{stimmt das Lemma}. Therefore, $h$ is an impact function in the sense of Definition \ref{temporary impact def}. Moreover, for $\varphi(x)=\phi(x)/x$, 
\begin{eqnarray}0\le\lim_{x\downarrow0}\frac{\varphi( \sqrt x)}{-h(-x)}&=&\lim_{x\downarrow0}\frac{\sqrt x\phi(\sqrt x)}{f(-x)}=\lim_{x\downarrow0}\frac{\sqrt x\phi(\sqrt x)}{f(x)}=\lim_{x\downarrow0}\frac{\big(\sqrt x\phi(\sqrt x)\big)'}{f'(x)}\nonumber\\&=&\lim_{x\downarrow0}\frac{\frac1{2\sqrt x} \phi(\sqrt x)+\frac12\phi'(\sqrt x)}{\sqrt{\phi'(\sqrt x)}}\le\lim_{x\downarrow0}\sqrt{\phi'(\sqrt x)}=0,\label{varphi/h}
\end{eqnarray}
where we have used l'Hopital's rule in the third step and, in the  inequality, the following estimate, which holds by the convexity of $\phi$:
$$\phi(\sqrt x)=\int_0^{\sqrt x}\phi'(y)\,dy\le \sqrt x\phi'(\sqrt x).
$$

Now we take $x\in(0,1]$, $\HX:=\sqrt x$,  and send $T$ to infinity in \eqref{costs of xi in proof of Thm 4.1} to get
\begin{eqnarray*}
\lim_{T\uparrow\infty}\mathbb E[\,\cC_T\,]&=&r f(-x)+\mathbb E\bigg[\,\phi(\sqrt x)-\gamma(1-\alpha)x^{3/2}\tau_1+\kappa h(-x)\sqrt x;\ \tau_1\le r\text{ and }\widetilde Y_1\ge\sqrt x\,\bigg]\\
&\le&-rxh(-x)+\phi(\sqrt x)+\kappa h(-x)\sqrt xr\lambda_0\lambda_1(1)\\
&=&-\sqrt xh(-x)\Big(r\sqrt x+\frac{\varphi(\sqrt x)}{-h(-x)}-\kappa r\lambda_0\lambda_1(1)\Big).
\end{eqnarray*}
By \eqref{varphi/h}, the term in parenthesis, and hence the entire right-hand side, becomes strictly negative when $x$ is sufficiently small, and so there is price manipulation. This shows that we must have $\kappa=0$.

Now we prove that we must have $\alpha=1$ when there is no price manipulation for any time horizon. To this end, we assume by way of contradiction that $\alpha<1$. We
fix an arbitrary impact function $h_1$ and let $h:=\varepsilon h_1$ for $\varepsilon>0$. We also set $x:=1$ and take $\HX\in(0,1]$. Sending $\varepsilon$ to zero and $T$ to infinity in \eqref{costs of xi in proof of Thm 4.1} yields that 
\begin{eqnarray*}
\lim_{T\uparrow\infty}\lim_{\varepsilon\downarrow0}\mathbb E[\,\cC_T\,]&=&\mathbb E\Big[\,\phi(\HX)-\gamma(1-\alpha)\HX\tau_1;\,\tau_1\le r\text{ and }\widetilde Y_1\ge\HX\,\Big]\\
&\le& \HX\Big(\varphi(\HX)-\gamma(1-\alpha)\mathbb E\big[\,\tau_1;\,\tau_1\le r\text{ and }\widetilde Y_1\ge1\,\big]\Big).
\end{eqnarray*}
But since $\varphi(\HX)\downarrow0$ as $\HX\downarrow0$, the term in parenthesis, and hence the entire right-hand side, will be strictly negative when $\HX$ is sufficiently small. Hence, we cannot have $\alpha<1$ when there is no price manipulation for all Almgren--Chriss models.

(c)$\Rightarrow$(a): Assume that $\alpha = 1$, $\kappa=0$, and that \eqref{beta inequality all T} holds. Note that 
\begin{equation}\label{Xtaui+Eq}
\int_0^{\tau_i} \xi_s \, ds - X_{\tau_i +} = - \sum_{j=1}^i Y_i - X_0. 
\end{equation}
With Lemma~\ref{revenuesT} we get for the costs of an admissible strategy $(\HX, \Xi,\rho, \xi)$
\begin{eqnarray*} 
\mathcal C^\chi_T & = & -\int_0^TX_t\,dP^0_t+\frac\gamma2\bigg(X_0+\sum_{i=1}^{N_\rho} Y_i\bigg)^2+\int_0^T f(\xi_t)\,dt+\sum_{i=1}^{N_{\rho}}Y_i\bigg(-\gamma \sum_{j=1}^i Y_i - \gamma X_0+\beta(Y_i)\bigg)\\
&=&-\int_0^TX_t\,dP^0_t+\frac\gamma2X_0^2+\int_0^T f(\xi_t)\,dt-\frac\gamma2\sum_{i=1}^{N_\rho} Y_i^2+\sum_{i=1}^{N_\rho} Y_i\beta(Y_i)\\
&\ge&-\int_0^TX_t\,dP^0_t+\frac\gamma2X_0^2+\int_0^T f(\xi_t)\,dt,
\end{eqnarray*}
where we have used \eqref{Delta>0 for alpha=1 Condition}  in the final step. Therefore,
\[ \mathbb E[\,\mathcal C_T\,] \ge  \frac{\gamma}{2} X_0^2 + \mathbb E\left[\int_0^T f(\xi_t)\,dt\right]\ge0. \]
This establishes (a).
\end{proof}

\begin{proof}[Proof of Proposition \ref{thm-no-opt-strategies}]
Let $X_0 \in \mathbb R$ and $\alpha=1, \beta =0$. Using  \eqref{Xtaui+Eq}  and Lemma \ref{revenuesT}, one finds that the costs of any admissible  strategy  $\chi=(\HX,\M, \xi, \rho)\in\mathcal X(X_0,T)$ are given by
\begin{eqnarray*}
 \mathcal C^\chi_T &=&- \int_0^T X_t \, dP^0_t + \frac{\gamma}{2} X_0^2 + \int_0^T f(\xi_t)\,dt -\frac{\gamma}{2} \sum_{i=1}^{N_\rho} Y_i^2 + \sum_{i=1}^{N_\rho} Y_i \kappa h(\xi_{\tau_i}).
\end{eqnarray*} 
Consider the following strategy with $\rho = \frac{T}{2}$, $\Xi=0$,    $\widehat X\neq0$, and
\[ \xi_t = \begin{cases}
0, & \textrm{ if } 0\le t\le\rho,\\
-2 \frac{X_0 + Z_\rho}{T}, & \textrm{ if } \rho<t\le T.
\end{cases}
\]
The expected costs of this strategy are
\begin{eqnarray*}
\mathbb E[\,\mathcal C_T\,] &=&  \frac{\gamma}{2} X_0^2 + \mathbb E\left[\, \frac{T}{2} f\left(-2 \frac{X_0 + Z_{{T}/2}}{T}\right)\,\right] - \frac{\gamma}{2} \mathbb E\bigg[\,\sum_{i=1}^{N_{{T}/2}} Y_i^2\,\bigg]\\
&=&  \frac{\gamma}{2}  X_0^2 - \mathbb E\left[ \,\left(X_0 + Z_{{T}/2} \right) h\left(-2 \frac{X_0 + Z_{{T}/2}}{T}\right)\,\right] - \frac{\gamma}{2} \mathbb E\bigg[\,\sum_{i=1}^{N_{{T}/2}} Y_i^2\,\bigg].
\end{eqnarray*}
Note that $|Z_{{T}/2}|$ is bounded by $|\widehat X|$ for all $T$ and that $Y_i$ is nonzero only as long as $|\sum_{j=1}^{i-1} \widetilde Y_j| < |\HX|$. Hence, with probability one, only finitely many $Y_i$ are nonzero. Therefore,  and by dominated convergence,
\[ \lim_{T\uparrow\infty} \mathbb E[\,\mathcal C_T\,]= \frac{\gamma}{2} X_0^2 - \frac{\gamma}{2} \mathbb E\bigg[\,\sum_{i=1}^\infty Y_i^2\,\bigg].\]
When sending $|\HX| $ to infinity, $\sum_{i=1}^\infty Y_i^2$ tends to infinity with probability one. Hence,  we can make  $ \lim_{T\uparrow\infty} \mathbb E[\,\mathcal C_T\,]$ arbitrarily negative. So we can find a sequence of strategies $\chi_n \in\mathcal X(X_0, T_n)$ such that $\mathbb E[\,\mathcal C_T^{\chi_n}\,]\le -n$ which implies the assertion.
\end{proof}

\begin{proof}[Proof of Proposition \ref{thm-noPMforsmallT}] In the first step, we show that there is no price manipulation when $\kappa=0$, $h(x)=\eta x$, and $T\le 2\eta/\gamma$. To this end, take any  admissible round trip $\chi=(\HX,\M, \xi, \rho)\in\mathcal X(0,T)$. 
Lemma \ref{revenuesT} and \eqref{Xtaui+Eq} yield that the costs of  $\chi$ are given by
\begin{eqnarray*}
\mathcal C^\chi_T& =& -\int_0^T X_t \, dP^0_t +\frac{\gamma}{2}\bigg(\sum_{i=1}^{N_\rho} Y_i\bigg)^2+\eta\int_0^T \xi_t^2\,dt
+\gamma \sum_{i=1}^{N_\rho} Y_i\bigg((1-\alpha)\int_0^{\tau_i} \xi_s\,ds-\alpha\sum_{j=1}^i Y_j\bigg)\\
&=&-\int_0^T X_t \, dP^0_t+\alpha\bigg(\eta\int_0^T \xi_t^2\,dt-\frac{\gamma}{2} \sum_{i=1}^{N_\rho} Y_i^2\bigg)\\
&&+(1-\alpha)\bigg(\eta\int_0^T \xi_t^2\,dt+\gamma \sum_{i=1}^{N_\rho} Y_i \int_0^{\tau_i} \xi_s \, ds + \frac{\gamma}{2} \bigg(\sum_{i=1}^{N_\rho} Y_i\bigg)^2 \bigg).
\end{eqnarray*}
Note that
$$\eta\int_0^T \xi_t^2\,dt\ge\frac\eta T\bigg(\int_0^T\xi_t\,dt\bigg)^2=\frac{\eta}{T} \bigg(\sum_{i=1}^{N_\rho} Y_i\bigg)^2\ge \frac{\eta}{T}\sum_{i=1}^{N_\rho} Y_i^2,
$$
where we have used Jensen's inequality in the first step and the fact that all $Y_i$ have the same sign in the third.
Thus,
\begin{equation} \label{newTlemmaPf1} \frac{\gamma}{2} \sum_{i=1}^{N_\rho} Y_i^2-\eta\int_0^T \xi_t^2\,dt \le \left(\frac{\gamma}{2} - \frac{\eta}{T} \right)\sum_{i=1}^{N_\rho} Y_i^2.
\end{equation}
Furthermore, let $\Gamma:= \sup_{t \in [0,T]} |\int_0^t \xi_s \, ds|$. Then, by Jensen's inequality,
$$\int_0^T\xi_t^2\,dt\ge \frac1T\bigg(\int_0^T\xi_t\,dt\bigg)^2\ge \frac{\Gamma^2}T.
$$
We can estimate
$$\eta\int_0^T \xi_t^2\,dt+\gamma \sum_{i=1}^{N_\rho} Y_i \int_0^{\tau_i} \xi_s \, ds  \ge \eta \frac{\Gamma^2}{T} - \gamma\sum_{i=1}^{N_\rho} |Y_i|  \Gamma.
$$
As a function of $\Gamma$, the right-hand side becomes minimal when
$$ \Gamma=\frac{\gamma T}{2\eta}\sum_{i=1}^{N_\rho} |Y_i|.
$$
Therefore, 
\begin{equation}\label{newTlemmaPf2}\eta \int_0^T \xi_t^2\,dt+\gamma \sum_{i=1}^{N_\rho} Y_i \int_0^{\tau_i} \xi_s \, ds + \frac{\gamma}{2} \bigg(\sum_{i=1}^{N_\rho} Y_i\bigg)^2
\ge \frac\gamma2\bigg(\sum_{i=1}^{N_\rho} |Y_i|\bigg)^2  \left(1-\frac{\gamma T}{2\eta}  \right).
\end{equation}
Combining (\ref{newTlemmaPf1}) and (\ref{newTlemmaPf2}) yields
\begin{align*}
\mathcal C_T &\ge -\int_0^T X_t \, dP^0_t - \alpha \left(\frac{\gamma}{2}  - \frac{\eta}{T} \right)\sum_{i=1}^{N_\rho} Y_i^2 +(1-\alpha) \frac\gamma2\bigg(\sum_{i=1}^{N_\rho} |Y_i|\bigg)^2  \bigg(1-\frac{\gamma T}{2\eta}  \bigg) \\
&\ge \int_0^T X_t \, dP^0_t
\end{align*}
for $T\le 2{\eta}/{\gamma}$ and we find  $ \mathbb E[\,\mathcal C_T\,] \ge 0$. This proves that there is no price manipulation for $T\le 2\eta/\gamma$.

In the second step of the proof, we show that for $\kappa=0$, $h(x)=\eta x$, $\alpha=1$ there is price manipulation when $T > 2 {\eta}/{\gamma}$. To construct a corresponding price manipulation strategy, take $\varepsilon\in(0,T)$ for which
\[ \frac{\gamma}{2} - \frac{\eta}{T-\varepsilon} > 0.\]
Now we consider a round trip $\chi$ with  $\rho = \tau_1 \wedge \varepsilon$, $\Xi=0$, arbitrary $\widehat X\neq0$, and 
\[ \xi_t = \begin{cases} -\frac{Y_1}{T-\varepsilon}, & \textrm{ if } t > \varepsilon \textrm{ and } \tau_1 \le \varepsilon, \\
0, & \textrm { otherwise. } \end{cases}
\]
 The expected costs of this strategy are 
\[ \mathbb E[\,\mathcal C^\chi_T\,] = \left(\frac{\eta}{T-\varepsilon}-\frac{\gamma}{2} \right) \mathbb E[\,Y_1^2; \tau_1 \le \varepsilon\,] < 0.\]
Hence, $\chi$ is a price manipulation strategy.
\end{proof}

\begin{proof}[Proof of Proposition \ref{infinite-DP-thm-small-T}]
 Let $T>0$ and fix $\widehat X$ such that
$$\frac{\gamma}{2} \HX^2 - \frac{T}{2} f\Big(  2\frac{\HX}{T}\Big)=\frac{\widehat X^2}2\bigg(\gamma-\frac{h(2\widehat X/T)}{\widehat X}\bigg)>0,
$$ 
which is possible due to the sublinearity of $h$. We take $\Xi=|\HX|$  and let
$$\sigma:=\inf\{\tau_k\,|\,Y_k=\HX\}=\inf\{\tau_k\,|\,\widetilde Y_k\ge\Xi\}.
$$
Recall that the  $(\widetilde Y_k)$ form an i.i.d. sequence of random variables, which are unbounded from above. Therefore, and by the standard thinning argument for Poisson processes, $\sigma$ is exponentially distributed with a parameter $\lambda\in(0,\theta]$.
Now we take $\rho = {T}/{2}$ and 
\[ \xi_t := \begin{cases}
0, & t \le \rho, \\
0, & t > \rho \textrm{ and } \sigma > \rho,\\
-{2\HX}/{T}, & t > \rho \textrm{ and } \sigma \le \rho.
\end{cases}
\]
The expected costs of the strategy $\chi=(\HX,\Xi,\xi,\rho)$ are
\begin{eqnarray*}
 \mathbb E[\,\mathcal C^\chi_T\,] & = & \mathbb E\bigg[\int_0^T f(\xi_t)\, dt\bigg]-\int_0^\rho \lambda e^{-\lambda t}\frac{\gamma}{2} \HX^2 \, dt \nonumber \\
 & = &(1-e^{-\lambda \rho})\bigg(  \frac{T}{2} f\Big(  2\frac{\HX}{T}\Big)-\frac{\gamma}{2} \HX^2\bigg)\\
&<&0.
\end{eqnarray*}
So there is price manipulation.\end{proof}

Now we prove the results pertaining to the assumptions that  $\alpha = \beta = \kappa=0$ and $h(\xi) = \eta \xi$. Under these conditions, Lemma~\ref{revenuesT} yields that for any admissible strategy $\chi=(\HX,\Xi,\xi,\rho)$,
\begin{equation}\label{Propo4.6and4.7RevenuesEq}
\mathcal C^\chi_T = - \int_0^T X_t \, dP^0_t + \frac{\gamma}{2} \bigg(X_0 + \sum_{i=1}^{N_\rho} Y_i \bigg)^2 + \eta \int_0^T \xi_t^2 \,dt + \sum_{i=1}^{N_\rho} Y_i \left(\gamma \int_0^{\tau_i} \xi_s \, ds \right). 
\end{equation}

\begin{proof}[Proof of Proposition \ref{thm-sunny-no-pm}]
Proof of (a):
Take $\rho = \frac{T}{2}$ and
\[ \xi_t = \begin{cases}
-\frac{\gamma}{2 \eta} \HX, & \textrm{ if } t \le \tau_1, t \le \rho,\\
0, & \textrm{ if } t > \tau_1, t \le \rho\\
- \frac{X_{\rho+}}{\rho}, & \textrm{ if } t > \rho,
\end{cases} 
\]
where $\widehat X$ will be specified later. 
 By \eqref{Propo4.6and4.7RevenuesEq}, we find that the costs of $\chi=(\HX,0,\xi,\rho)$ are 
\[ \mathcal C^\chi_T = - \int_0^T X_t \, dP^0_t + \frac{\gamma}{2} \bigg(X_0 + \sum_{i=1}^{N_\rho} Y_i \bigg)^2 
+ \eta ( \rho\wedge \tau_1 ) \frac{\gamma^2}{4 \eta^2} \HX^2 + \eta \frac{X_{\rho+}^2}{\rho} - \sum_{i=1}^{N_\rho} Y_i \tau_1 \frac{\gamma^2}{2 \eta} \HX.\]
In the limit $T \uparrow \infty$  we will have 
$$
\sum_{i=1}^{N_\rho} Y_i=\sum_{i=1}^{N_{T/2}} Y_i\longrightarrow \widehat X.
$$
Hence, using the fact that $\mathbb E[\,\tau_1\,] = \frac{1}{\theta}$,
\begin{equation}\label{thm-sunny-no-pm-pf1}
\lim_{T\uparrow\infty} \mathbb E[\,\mathcal C^\chi_T\,] = \frac{\gamma}{2} (X_0 + \HX)^2 - \frac{1}{\theta} \frac{\gamma^2}{4 \eta} \HX^2.
\end{equation}
Choosing 
$$\HX = - \frac{2 X_0 \eta \theta}{\gamma - 2 \eta \theta}$$ yields
\[ \lim_{T\uparrow\infty}\mathbb E[\,\mathcal C^\chi_T\,] = -\frac{\gamma}{2}^2 X_0^2 \frac{1}{2 \eta \theta - \gamma}<0. \]
This concludes the proof of part (a).

Proof of (b): We first consider the case in which $\frac{\gamma}{\eta \theta} = 2$ and $X_0 \ne 0$. With the same strategy $\chi$ as in part (a) we find with (\ref{thm-sunny-no-pm-pf1}) that
\[ \lim_{T\uparrow\infty} \mathbb E[\,\mathcal C^\chi_T\,] = \frac{\gamma}{2} X_0^2 + \gamma X_0 \HX.\]
For $X_0\neq0$, the right-hand side can be made arbitrarily negative  by taking $\widehat X$ with the opposite sign of $X_0$ and making $|\widehat X|$ large. 

Now we consider the case in which  $\frac{\gamma}{\eta \theta} > 2$. With (\ref{thm-sunny-no-pm-pf1}) we find 
\[ \lim_{T\uparrow\infty} \mathbb E[\,\mathcal C^\chi_T\,] =  \frac{\gamma}{2} X_0^2 + \gamma X_0 \HX - \varepsilon \HX^2,\]
where $\varepsilon > 0$. Again, the right-hand side can be made arbitrarily negative by sending $\widehat X$ to infinity.
\end{proof}

\begin{proof}[Proof of Proposition \ref{trans-trigg-lemma}] In view of Proposition \ref{thm-sunny-no-pm}, the assertion will be implied by the following claim:  If, for  $0 \le t < \rho$, we have $\xi_t\le 0$ when  $X_0>0$ or $\xi_t\ge0$ when $X_0<0$, then
$\mathbb E[\,\mathcal C_T\,] \geq 0$.
In  proving this claim, we will consider the case $X_0>0$. The case $X_0<0$ is analogous.
With Lemma~\ref{revenuesT} we find that for any admissible strategy $\chi=(\HX,\Xi,\xi,\rho)\in\mathcal X(X_0,T)$,
\[ \mathcal C^\chi_T = - \int_0^T X_t \, dP^0_t + \frac{\gamma}{2} \bigg(X_0 + \sum_{i=1}^{N_\rho} Y_i \bigg)^2 + \int_0^T f(\xi_t) \, dt + \sum_{i=1}^{N_\rho} Y_i \gamma \int_0^{\tau_i} \xi_s \, ds. \]
Consider first the case $\HX \le 0$. Then 
\[  \sum_{i=1}^{N_\rho} Y_i \gamma \int_0^{\tau_i} \xi_s \, ds \ge 0 \]
and $\mathbb E[\,\mathcal C^\chi_T\,] \geq 0$ follows.

Consider next the case $\HX > 0$. Since $\xi_t \le 0$ this implies $X_t \ge 0$ for all $t$. Especially, $X_{\tau_i -} \ge 0$, or equivalently, 
\[ \int_0^{\tau_i} \xi_s \, ds \ge - X_0 - \sum_{j=1}^{i-1} Y_i. \]
Therefore, we find that
\[ \mathcal C^\chi_T \ge - \int_0^T X_t \, dP^0_t + \frac{\gamma}{2} X_0^2 + \frac{1}{2}{\gamma} \sum_{i=1}^{N_\rho} Y_i^2 + \int_0^T f(\xi_t) \, dt \]
and $\mathbb E[\,\mathcal C^\chi_T\,] \geq 0$ follows.\end{proof}

\begin{proof}[Proof of Proposition \ref{infinite-DP-thm-sunny-no-pm}]
Under the assumption $\mu=\delta_\infty$, the entire dark-pool order $\HX$ is is filled at $\tau_1$. Therefore,  the costs of a strategy $\chi\in\mathcal X(X_0,T)$ are
\begin{eqnarray*}
\mathcal C^\chi_T & =& - \int_0^T X_t \, dP^0_t +\frac{\gamma}{2} (X_0 + \indf{\tau_1 \le \rho} \HX)^2 +\eta \int_0^T \xi_t^2 \, dt + \indf{\tau_1 \le \rho} \gamma \HX (X_{\tau_1 -} - X_0) \\
& \ge & - \int_0^T X_t \, dP^0_t +\indf{\tau_1\le\rho} \left(\frac{\gamma}{2} (X_0 + \HX)^2 + \eta \frac{(X_{\tau_1-} - X_0)^2}{\tau_1} + \gamma \HX (X_{\tau_1-} - X_0) \right), 
\end{eqnarray*}
where we have applied Jensen's inequality to the integral $ \int_0^{\tau_1} \xi_t^2 \, dt$, which is less than or equal to $ \int_0^T \xi_t^2 \, dt $.
As a function of $X_{\tau_1-} $, the rightmost expression is minimized by
\[ X_{\tau_1-} = X_0 - \frac{\gamma}{2\eta} \tau_1 \HX, \]
and so
\begin{equation}\label{Prop6.2AuxIneq1}
\mathcal C^\chi _T \ge - \int_0^T X_t \, dP^0_t +\indf{\tau_1<\rho} \left(\frac{\gamma}{2} (X_0 + \HX)^2 - \frac{\gamma^2}{4 \eta} \tau_1 \HX^2 \right).
\end{equation}
Thus,
\[ \mathbb E[\,\mathcal C^\chi_T\,] \ge  \mathbb E[\,\psi(\rho, \widehat X)\,], \]
where
\begin{eqnarray*}\psi(\rho, \widehat X) := \int_0^\rho \theta e^{-\theta t} \left(\frac{\gamma}{2} (X_0 + \HX)^2 - \frac{\gamma^2}{4 \eta} t \widehat X^2 \right) \,dt.
\end{eqnarray*}
We see that $\psi(0, \widehat X) = 0$ and the term in parentheses is decreasing in $t$ for $\HX\neq0$. Therefore, if $\widehat X$ is such that $\psi (\infty, \widehat X) < 0$, then we have $\psi (\rho, \widehat X) \ge\psi (\infty, \widehat X)$ for all $\rho < \infty$. For $\widehat X$ with $\psi (\infty, \widehat X) \ge 0$ we have $\psi (\rho, \widehat X)\ge 0$ for all $\rho < \infty$. Thus,
\[ \mathbb E[\,\mathcal C^\chi_T\,] \ge 0\wedge \psi(\infty, \widehat X) = -\Big(  \frac{\gamma}{2} (X_0 + \widehat X)^2 - \frac{\gamma^2}{4 \eta \theta} \widehat X^2\Big)^-. \]
The right-hand side is minimized by taking
\[ \HX = �2 X_0 \frac{\eta \theta}{\gamma - 2 \eta \theta}, \]
and so
\[\mathbb E[\,\mathcal C^\chi_T\,] \ge- \frac{\gamma^2}{2}  X_0^2 \frac{1}{2 \eta \theta - \gamma}.\]
The assertion now follows with part (a) of Proposition \ref{thm-sunny-no-pm}.
\end{proof}

\begin{proof}[Proof of Corollary \ref{CounterExCor}] We already know from Proposition \ref{thm-sunny-no-pm} (b) that there is price manipulation for $\frac{\gamma}{\eta} > 2 \theta$. On the other hand, Proposition \ref{infinite-DP-thm-sunny-no-pm} implies that there is no price manipulation for $\frac{\gamma}{\eta} < 2 \theta$. Hence, it remains to analyze the case $\frac{\gamma}{\eta} = 2 \theta$.
 For a round trip $\chi\in\mathcal X(0,T)$, our estimate \eqref{Prop6.2AuxIneq1} yields that 
\begin{eqnarray*} \mathcal C^\chi_T &\ge&  -\int_0^T X_t \, dP^0_t +\indf{\tau_1<\rho} \gamma\left(\frac{1}{2}-\frac{\gamma}{4 \eta} \tau_1     \right)\HX^2 .
\end{eqnarray*} 
Hence,
$$\mathbb E[\,\mathcal C^\chi_T\,] \ge \gamma\widehat X^2\mathbb E[\,g(\rho)\,]
$$
where
$$g(\rho):=\int_0^\rho \theta e^{-\theta t} \left(\frac{1}{2}-\frac{\gamma}{4 \eta} t     \right) \,dt=\frac12(1-e^{-\theta\rho})-\frac{\gamma}{4 \eta\theta} \big(1-e^{-\theta\rho}(1+\theta\rho) \Big)=\frac12\theta\rho e^{-\theta \rho}\ge0.
$$
This gives $\mathbb E[\,\mathcal C^\chi_T\,] \ge0$.
\end{proof}

\begin{proof}[Proof of Proposition \ref{thm-semi-deterministic}]
Under the assumptions $\alpha=1$, $\beta(y)=\frac\gamma2y$, and $\kappa=0$, 
 the costs of an admissible strategy $\chi\in\mathcal X(X_0,T)$ are given by 
\begin{eqnarray*}
 \mathcal C^\chi_T &=& - \int_0^T X_t \, dP^0_t +\frac{\gamma}{2}\left(X_0+\mathbbmss 1_{\{\tau_1\le\rho\}}\HX\right)^2 + \int_0^T f(\xi_t)\, dt- \gamma\HX \left(  \frac12\HX + X_0 \right) \indf{\tau_1\le\rho}\\
 & = & - \int_0^T X_t \, dP^0_t +\frac{\gamma}{2} X_0^2 + \int_0^T f(\xi_t)\,dt.
\end{eqnarray*}
Taking the conditional expectation with respect to $\mathcal{F}_{{\tau_1}\wedge\rho}$ and using optional sampling yields
\begin{align*}\mathbb E[\,\mathcal C^\chi_T\,|\,\mathcal{F}_{{\tau_1}\wedge\rho}\,]&=-\int_0^{{\tau_1}\wedge\rho}X_t\,dP^0_t+\frac\gamma2X_0^2 + \int_0^{{\tau_1}\wedge\rho}f(\xi_t)\,dt +  \mathbb E\Big[\,\int_{{\tau_1}\wedge\rho}^Tf(\xi_t)\,dt\,\big|\,\mathcal{F}_{{\tau_1}\wedge\rho}\,\Big].
\end{align*}
Due to the liquidation constraint, we must have  $\int_{{\tau_1}\wedge\rho}^T\xi_t\,dt=-X_{{\tau_1}\wedge\rho}-\mathbbmss 1_{ \{ {\tau_1} <\rho \}}\widehat X$, and so the convexity of $f$ and Jensen's inequality yield that
$$\int_{{\tau_1}\wedge\rho}^Tf(\xi_t)\,dt\ge(T-{\tau_1}\wedge\rho)f\Big(\frac{-X_{{\tau_1}\wedge\rho}-\mathbbmss 1_{ \{ {\tau_1} <\rho \}}\widehat X}{T-{\tau_1}\wedge\rho}\Big)
$$
with equality if, for ${\tau_1}\wedge\rho\le t\le T$,
\begin{equation}\label{xiNachtauwedgerhoEq}
\xi_t=\begin{cases}\displaystyle\frac{-X_{\tau_1}-\widehat X}{T-{\tau_1}}&\text{on }\{ {\tau_1} \le\rho \},\\ &\\
\displaystyle\frac{-X_\rho}{T-\rho}&\text{on }  \{ \rho<{\tau_1} \}.
\end{cases}
\end{equation}
These two possibilities will correspond to the single update of the optimal strategy at ${\tau_1}$.

Note next that, due to the $(\mathcal G_t)$-predictability of the processes $(\xi_t)$ and $(\rho\wedge t)_{t\ge0}$, $(\xi_s)_{s\le t}$ and $\rho\wedge t$ are independent of ${\tau_1}$, conditional on $\{t\le{\tau_1}\}$.
It follows that 
\begin{align}\mathbb E[\,\mathcal C^\chi_T\,]&=\mathbb E[\,\mathbb E[\,\mathcal C^\chi_T\,|\,\mathcal{F}_{{\tau_1}\wedge\rho}\,]\,]\nonumber\\
&\ge \frac\gamma2X_0^2 +\mathbb E\big[\, \int_0^{{\tau_1}\wedge\rho}f(\xi_t)\,dt+(T-{\tau_1}\wedge\rho)f\Big(\frac{-X_{{\tau_1}\wedge\rho}-\mathbbmss 1_{ \{ {\tau_1}\le \rho \}}\widehat X}{T-{\tau_1}\wedge\rho}\Big)\,\bigg]\nonumber\\
&= \frac\gamma2X_0^2 +\mathbb E\big[\, F(\widehat X, \xi, \rho)\,\big],\label{SingleUpdateRevenuesEq}
\end{align}
where the functional $F$  maps $\widehat X\in\mathbb R$,   $\xi\in L^1[0,T]$, and $r\in[0,T]$  to
\begin{align*}
F(\widehat X, \xi,r)=&\int_0^{\infty}du\,\theta e^{-\theta u}\bigg\{\int_0^{u\wedge r }f(\xi_t)\,dt+(T-u\wedge r)f\Big(\frac{-X_0-\int_0^{u\wedge{r}}\xi_t\,dt-\mathbbmss 1_{ \{ u\le{r} \}}\widehat X}{T-u\wedge{r}}\Big)\bigg\}.
\end{align*}
When $F$ admits a minimizer $(\widehat X^*,\xi^*,r^*)$, then  concatenating $\xi^*$ with \eqref{xiNachtauwedgerhoEq} in $r^*\wedge{\tau_1}$ yields an optimal strategy that is a single-update strategy. 

To show the existence of a minimizer of $F$, take any triple $(\widetilde X,\widetilde \xi,\widetilde r)$ for which $C:=F(\widetilde X,\widetilde \xi,\widetilde r)<\infty$. We then only need to look into those triples $(\widehat X, \xi,r)$  for which $F(\widehat X, \xi,r)\le C$.   Without loss of generality, we can pick the component $\xi$ from the set  
$$K_C:=\Big\{\xi\in L^1[0,T]\,\Big|\,\int_0^Tf(\xi_t)\,dt\le {C}e^{\theta T}\Big\},
$$
because 
$$F(\widehat X, \xi,r)\ge \int_T^{\infty}du\,\theta e^{-\theta u}\int_0^{u\wedge r }f(\xi_t)\,dt=e^{-\theta T}\int_0^rf(\xi_t)\,dt
$$
and we can set $\xi_t:=0$ for $t>r$.  

The set $K_C$ is a  closed convex subset of $L^1[0,T]$.  Hence it is also weakly  closed in $L^1[0,T]$. It is also uniformly integrable according to the criterion of de la Vall\'ee Poussin and the superlinear growth of $f$, which follows from our assumption that $|h(x)|\to\infty$ as $|x|\to\infty$. Hence, the Dunford--Pettis theorem  \cite[Corollary IV.8.11]{DunfordSchwartz} implies that $K_C$ is  weakly sequentially compact in $L^1[0,T]$. From now on we will endow $K_C$ with the weak topology.

It follows in particular that 
\begin{equation}\label{supL1KCfiniteEq}
\sup_{\xi\in K_C}\int_0^T|\xi_t|\,dt<\infty.
\end{equation}
Since
$$F(\widehat X, \xi,r)\ge\int_0^r du\,\theta e^{-\theta u}(T-u)f\Big(\frac{-X_0-\int_0^{u}\xi_t\,dt-\widehat X}{T-u}\Big),
$$
the superlinear growth of $f$ and \eqref{supL1KCfiniteEq} imply that there is a constant $C_1\ge0$ such that $|\widehat X|\le C_1$
when $F(\widehat X, \xi,r)\le C$. Hence we can restrict our search of a minimizer to the sequentially compact set
$$\mathcal K:=[-C_1,C_1]\times K_C\times[0,T].
$$

Next, 
$$[0,T]\times K_C\ni(r,\xi)\longrightarrow\int_0^r\xi_t\,dt=\int_0^T\xi_t\mathbbmss 1_{[0,r]}(t)\,dt
$$
is a continuous map. Moreover, denoting by $f^*$ the Fenchel-Legendre transform of the finite convex function $f$, we have $f^{**}=f$ due to the biduality theorem, and so 
$$[0,T]\times K_C\ni(r,\xi)\longmapsto\int_0^rf(\xi_t)\,dt=\sup_{\varphi\in L^\infty}\bigg[\int_0^T\mathbbmss 1_{[0,r]}(t)\xi_t\varphi_t\,dt-\int_0^rf^*(\varphi_t)\,dt\bigg];
$$
see, e.g., Theorem 2 in \citet{rockafellar68}. It follows that this map is lower semicontinuous  as the supremum of continuous maps. 

Altogether, it follows  that $F$ is lower semicontinuous on the sequentially compact set $\mathcal K$ and so admits a minimizer.
\end{proof}

\parskip-0.5em\renewcommand{\baselinestretch}{0.8}\small
\bibliography{literature}{}
\bibliographystyle{natbib}

\end{document}